\documentclass[twoside,12pt,a4paper,english]{article}
\usepackage{amssymb,amsxtra}
\usepackage{array}
\usepackage{graphics}            %
\usepackage[dvips]{graphicx} 
\usepackage{sidecap}
\usepackage{latexsym}
\usepackage[T1]{fontenc}
\usepackage[active]{srcltx}
\usepackage{amsthm}
\newtheorem{theorem}{Theorem}
\newtheorem{corollary}{Corollary}
\newtheorem{proposition}{Proposition}

\def\centeron#1#2{{\setbox0=\hbox{#1}\setbox1=\hbox{#2}\ifdim
    \wd1>\wd0\kern.5\wd1\kern-.5\wd0\fi
    \copy0\kern-.5\wd0\kern-.5\wd1\copy1\ifdim\wd0>\wd1
    \kern.5\wd0\kern-.5\wd1\fi}}

\newcommand{\mc}{\mathcal}

\newcommand{\N}{\mathbb}

\newcommand{\be}{\begin{equation}}
\newcommand{\ee}{\end{equation}}
\newcommand{\ba}{\begin{eqnarray}}
\newcommand{\ea}{\end{eqnarray}}
\newcommand{\nn}{\nonumber}
\newcommand\bigX{\makebox(-10,-10){\text{\Huge $X$}}}
\newcommand\bigXdown{\makebox(15,0)[dl]{\text{\Huge $X$}}}
\title{Conditions for the custodial symmetry in multi-Higgs-doublet models}
\author{M.~Aa.~Solberg\\
Department of Mechanical and Industrial Engineering,\\ NTNU, 7491 Trondheim, Norway. \\
E-mail: marius.solberg@ntnu.no 
}

\begin{document}
\date{}
\maketitle
\begin{abstract}
We derive basis-independent, necessary and sufficient conditions for the custodial symmetry in N-Higgs-doublet models
(NHDM) for $N\geq 3$, and apply them on some 3HDM examples.
\end{abstract}
\newpage
\tableofcontents
\newpage

\section{Introduction}
\label{sec:Intro}
The custodial 
symmetry in the standard model (SM) is an approximate symmetry which guards
the $\rho$ parameter
\ba
    \rho =\frac{m_W^2}{m_Z^2\cos^2(\theta_W)},
\ea
from large radiative corrections \cite{Sikivie:1980hm}. Here $m_W$ and $m_Z$ are the masses of the $W$ and $Z$ bosons, while $\theta_W$ is the weak mixing angle.  At the first order of perturbation theory (tree-level), we have 
\ba\label{E:rho1}
\rho=1, 
\ea
since $\cos^2(\theta_W)=m_W^2/m_Z^2$ at tree-level.
Eq.~\eqref{E:rho1} holds at all orders of perturbation theory when the custodial symmetry is exact. The $\rho$ parameter is measured experimentally 
to be close to unity, with a minute dependence of the chosen renormalization prescription.
For instance, if we interpret $\rho$ in the
minimal subtraction renormalization scheme at the the energy scale $m_Z$,
$\rho$ takes the value $\rho=1.01032 \pm 0.00009$ 
\cite{Patrignani:2016xqp}.

The custodial symmetry is a $SO(4)$ symmetry of the Higgs sector, broken down to $SO(3)$ by spontaneous symmetry breaking. In the SM, the Higgs potential can easily be seen to be $O(4)$-symmetric. 
Let the SM Higgs Lagrangian be given by
\ba\label{E:SMlag}
  \mc{L_H}= (D_\alpha \Phi)^\dag(D^\alpha \Phi) -V_\text{SM}(\Phi),
\ea
where 
\ba
\Phi = \begin{pmatrix} \phi^+ \\ \phi^0   \end{pmatrix}=\begin{pmatrix} \phi_1 +i\phi_2 \\ v+\phi_3+i\phi_4   \end{pmatrix}
\ea 
is the SM Higgs doublet, and the fields $\phi_1,\ldots,\phi_4$ are real scalar fields and $v$ is the vacuum expectation value (VEV).
Moreover, in \eqref{E:SMlag} $D^\alpha$ is the covariant  derivative 
\ba
    D^\alpha =\partial^\alpha+ \frac{ig}{2}\sigma_j W_j^\alpha+\frac{ig'}{2} B^\mu,
\ea
$V(\Phi)$ is the Higgs potential
\ba\label{E:SMpot}
    V_\text{SM}(\Phi)= \lambda (\Phi^\dag \Phi)^2 +\mu^2 \Phi^\dag \Phi,
\ea
  and $\lambda$ and $\mu$ are constants. 
If we organize the Higgs doublet $\Phi$ as a real quadruplet, 
\ba\label{E:realQuadrSM}
    \Phi_r=\begin{pmatrix} \text{Re}(\Phi) \\  \text{Im}(\Phi) \end{pmatrix},
\ea
 the potential $V_\text{SM}(\Phi_r)$, given by \eqref{E:SMpot} with the substitution $\Phi \to \Phi_r$, is evidently invariant under  transformations 
\ba
\Phi_r\to O\Phi_r,
\ea
with $O\in O(4)$. The kinetic terms of \eqref{E:SMlag}
will, in the limit $g'\to 0$, be invariant under $SO(4)$ transformations  \cite{Olaussen:2010aq}, but not under $O\in O(4)$ with $\det(O)=-1$ \cite{Solberg:2012au}. This approximate $SO(4)$ symmetry is
the custodial $SO(4)$ symmetry, which we will denote $SO(4)_C$. In the presence of a VEV, this symmetry is broken down to a custodial $SO(3)$ symmetry,
denoted by $SO(3)_C$.

If we impose the custodial symmetry on the SM Lagrangian, it forces $\rho=1$ at all levels of perturbation theory. To see this, consider the mass-squared matrix $\mc{M}^2$ of the electroweak gauge bosons $W^\pm$ and $Z^0$, 
\ba\label{E:MassMsqWZ}
   \mc{M}^2= \frac{v^2}{2}\begin{pmatrix} g^2 & 0&0 &0 \\0 &g^2 &0 &0 \\0 &0 &g^2 &-g g' \\
   0&0&-gg'&g'^{2}  \end{pmatrix},
\ea
  where $g$ is the weak isospin coupling and $g'$ is the $U(1)_Y$ hypercharge coupling.
When the custodial symmetry is enforced on the kinetic Higgs terms the hypercharge coupling must be zero, $g'\to 0$, and then the three massive gauge bosons will transform as a triplet under custodial $SO(3)_C\subset SO(4)_C$, where $SO(3)_C$ here leaves the vacuum invariant (it leaves the VEV alone).
Then $m_W=m_Z$ at all orders of perturbation theory, since the massive gauge boson fields can be interchanged
by a $SO(3)_C$ transformation. If the custodial symmetry is extended to the Yukawa sector, the mass renormalization of the gauge bosons, due to massive fermions, will yield the same result for all the massive gauge bosons, and hence the mass degeneration of $W^\pm$ and $Z^0$ will still be exact \cite{Willenbrock:2004hu}.  
 Moreover, $\theta_W=0 $ at all orders since $g'=0$ implies no electroweak mixing. Hence $\cos^2\theta_W=1$, which gives us $\rho=1$ at all orders, when the theory is custodially symmetric.       
\subsection{N-Higgs-doublet models}
\label{sec:NHDM}
Augmentations of the scalar sector to include several Higgs doublets can be applied to implement $CP$ violation  
and dark matter. Higgs doublets generally respect the custodial symmetry, except for certain combinations
of the doublets associated with complex parameters. Given $N$ Higgs doublets $\Phi_1,\ldots,\Phi_N$,
we will construct the most general NHDM potential from the following Hermitian bilinears, linear in each field $\Phi_m, \Phi_n$:
\ba\label{E:bilinears}
    \widehat{A}_m &=& \Phi_m^\dag \Phi_m, \nn \\
  \widehat{B}_{mn} &=& \frac{1}{2}(\Phi_m^\dag \Phi_n+\Phi_n^\dag \Phi_m) \equiv \widehat{B}_a, \nn \\
\widehat{C}_{m n}&=& \frac{-i}{2}(\Phi_m^\dag \Phi_n-\Phi_n^\dag \Phi_m)\equiv \widehat{C}_a.
\ea  
To avoid double counting, we will let $1\leq m<n \leq N$, and we apply the following invertible encoding $a=a(m,n)$ to label such pairs,
\ba\label{E:encoding}
 1 \leq a(m,n) =  (n-1) + (m-1)\left(N-\frac{1}{2}(m+2)\right)\leq \frac{1}{2} N(N-1)\equiv k.
\ea  
The inverse of this encoding will be denoted $(m(a),n(a))$.
The bilinears $\widehat{B}_{m n}$ and $\widehat{C}_{m n}$ will be ordered "lexicographically" in $(m,n)$ by the encoding $a(m,n)$, e.g.
\ba
\{\widehat{C}_a\}_{a=1}^k=\{\widehat{C}_{12},\widehat{C}_{13},\ldots,\widehat{C}_{1 N},\widehat{C}_{23},\widehat{C}_{24},\ldots, \widehat{C}_{N-1,N}\}.
\ea
The most general NHDM potential can then be written,
\begin{align}\label{E:GeneralPotNHDM}
 V(\Phi_1, \ldots,\Phi_N) &= \mu_m^{(1)}\widehat{A}_m
 +\mu_{a}^{(2)}\widehat{B}_{a} + \mu_{a}^{(3)}\widehat{C}_{a}
 + \lambda_{mn}^{(1)}  \widehat{A}_m\widehat{A}_n +
 \lambda_{ab}^{(2)}\widehat{B}_{a}\widehat{B}_{b}\nonumber \\ 
  &+ \lambda_{ab}^{(3)} \widehat{C}_{a}\widehat{C}_{b} 
 + \lambda_{ma}^{(4)} \widehat{A}_{m}\widehat{B}_{a} 
 + \lambda_{ma}^{(5)}  \widehat{A}_{m}\widehat{C}_{a}
 + \lambda_{ab}^{(6)} \widehat{B}_{a} \widehat{C}_{b},
\end{align}
where repeated indices $m$ and $n$ are summed from $1$ to $N$, while repeated indices $a$ and $b$ are summed from $1$ to $k$, confer \eqref{E:encoding}.
The general NHDM Lagrangian can then be written
\ba\label{E:NHDMlag}
  \mc{L}_\text{NHDM}= (D_\alpha \Phi_m)^\dag(D^\alpha \Phi_m) - V(\Phi_1, \ldots,\Phi_N),
\ea
where $V$ was given in \eqref{E:GeneralPotNHDM}, and where we as usual sum over repeated indices.
However, the general NHDM Lagrangian has more sources of custodial symmetry violation than just
$g'$. The bilinears $\widehat{C}_{a}$ all violate the custodial symmetry, and hence the parameters
$\mu_{a}^{(3)}, \lambda_{ab}^{(3)}, \lambda_{ma}^{(5)}$ and $\lambda_{ab}^{(6)}$ will give contributions to $\Delta \rho = \rho -1$.
Consider transformations of the real NHDM quadruplets, given by 
\ba\label{E:realQuadrNHDM}
   \Phi_{m,r} = \begin{pmatrix} \text{Re}(\Phi_m) \\  \text{Im}(\Phi_m) \end{pmatrix} \to \Phi_{m,r}' =S  \Phi_{m,r},
\ea 
 as in \eqref{E:realQuadrSM} for the SM.
In section 2.2 of \cite{Olaussen:2010aq} we showed that the transformations $S$ which left the bilinears $\widehat{C}_{a}$ invariant, were the transformations $S\in Sp(2,\N{R})$, i.e.~the symmetry group of the bilinears $\widehat{C}_{mn}$ is the real symplectic group $Sp(2,\N{R})$. The symmetry group of $\widehat{C}_{a}\widehat{C}_{b}$ was the group $P(2,\N{R})$, with $Sp(2,\N{R})$ as identity component. 
Both symmetry groups are incompatible with the custodial symmetry, see figure \ref{figmulti}, and hence violate the custodial symmetry.
\begin{figure*}
\begin{center}
\includegraphics[width=1\textwidth]
{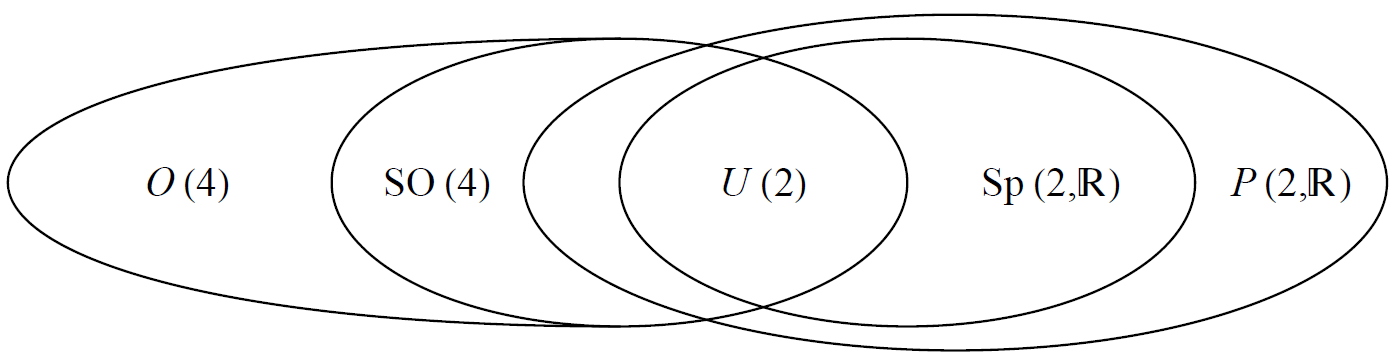}%
\end{center}
\caption{Diagram showing the incompatibility between the custodial symmetry and the symmetry groups of bilinears $\widehat{C}_{a}$ and quartic terms $\widehat{C}_{a}\widehat{C}_{b}$. The $SO(4)$ symmetry group is the custodial symmetry, $U(2)\cong SO(4)\cap Sp(2,\N{R}) \cong SU(2)_L\times U(1)_Y$ is the global symmetry of the SM, $Sp(2,\N{R})$ is the symmetry group of bilinears $\widehat{C}_a$, while 
$P(2,\N{R})$ is the symmetry group of quartic terms $\widehat{C}_a\widehat{C}_b$. Finally, $O(4)$ is the symmetry group of the bilinears $\widehat{B}_a$. Taken from \cite{Solberg:2012au}.
\label{figmulti}}
\end{figure*}
The most general custodially symmetric NHDM potential, is then given by
\begin{align}\label{E:potNHDMSO4inv}
 V_\text{CS}(\Phi_1, \ldots,\Phi_N) &= \mu_m^{(1)}\widehat{A}_m
 +\mu_{a}^{(2)}\widehat{B}_{a}
 + \lambda_{mn}^{(1)} \widehat{A}_m\widehat{A}_n
 + \lambda_{ab}^{(2)}\widehat{B}_{a}\widehat{B}_{b}\nonumber\\
 &+ \lambda_{ma}^{(4)} \widehat{A}_{m}\widehat{B}_{a}.
\end{align}
We will denote a NHDM potential of the form \eqref{E:potNHDMSO4inv}  \emph{manifestly $SO(4)_C$-symmetric}.
 However, since the Higgs doublets $\Phi_1,\ldots,\Phi_N$ have the same quantum numbers,
we are free to redefine the Higgs doublets through unitary \emph{Higgs basis transformations},
\ba\label{E:HbasisTrafo}
    \Phi_m\to \Phi_m'=U_{mn}\Phi_n,
\ea
for an $U\in U(N)$. The kinetic terms of the NHDM Lagrangian are invariant under $SU(2)_L \times U(N)$ transformations (promoted to $SU(2)_L \times Sp(N)$ in the limit $g'\to 0$ \cite{Olaussen:2010aq}),
and the basis transformations $U\in U(N)$ will thus leave the kinetic terms invariant. Normally only Higgs basis transformations $U\in SU(N)$ are considered, since an overall $U(1)$ transformation do not change the parameters of the potential. However, the NHDM potential will generally not be invariant under a $SU(N)$ change of Higgs basis, and hence the custodial symmetry can be hidden due to a change of basis. On the other hand, a NHDM potential that is not manifestly $SO(4)_C$-symmetric may be transformed into a manifestly $SO(4)_C$-symmetric potential through a $SU(N)$ basis transformation. Hence we define a potential to be \emph{explicitly} custodially symmetric, or more simple, \emph{$SO(4)_C$-symmetric},
if it can be transformed into a manifestly $SO(4)_C$-symmetric potential by a Higgs basis transformation. Thus a $SO(4)_C$-symmetric potential can be transformed to the form \eqref{E:potNHDMSO4inv}. On very general grounds, an implementation of the custodial symmetry
in a NHDM potential can be transformed to a manifest $SO(4)_C$ symmetry, i.e.~the potential can be transformed to the form \eqref{E:potNHDMSO4inv}, by a Higgs $SU(N)$ basis transformation \cite{Nishi:2011gc}.  

We will in this article develop necessary and sufficient conditions for the custodial symmetry in arbitrary NHDM potentials for $N>2$. For $N=2$, necessary and sufficient conditions for $SO(4)_C$ symmetry were derived independently in \cite{Grzadkowski:2010dj} and \cite{Haber:2010bw}. Necessary and sufficient conditions for $SO(4)_C$ in the case $N=3$, and necessary conditions
for the cases $N=4$ and $N=5$ are given in \cite{Nishi:2011gc}, in a different formalism than the one applied in this article.
 The cases $N\geq 3$ are different from the case $N=2$, since the bilinears \eqref{E:bilinears} transform under the adjoint representation $Ad_{SU(N)}$ of $SU(N)$, and $Ad_{SU(2)}=SO(3)$ while 
\ba
Ad_{SU(N)} 
\subsetneq  SO(N^2-1) \quad \text{for} \quad N>2.
\ea
Hence most $SO(N^2-1)$ matrices will not be at our disposal for $N>2$, and this will complicate the procedure, when we try to rotate a possibly $SO(4)_C$-symmetric potential into a manifestly $SO(4)_C$-symmetric potential.
\section{Basis-independent conditions for custodial symmetry in NHDM}
\label{sec:BasisIndependentConditionsForCustodialSymmetryInNHDM}
For a discussion of the special case where the quartic terms can be factorized, see appendix \ref{sec:FactorizableQuarticTerms}.

\subsection{The general NHDM potential}
\label{sec:TheGeneralNHDMPotential}
We will now find basis-independent, sufficient and necessary conditions for having a custodially symmetric potential in the general NHDM. To do this, we will adopt much of the notation applied in \cite{Maniatis:2015gma}. 

Now define
\begin{align}
	\vec{\Phi}=(\Phi_1,\Phi_2,\ldots,\Phi_N)^T,
\end{align}
and let the Hermitian $N\times N$ matrix $\tilde{K}$ be given by
\begin{align}
	\tilde{K} = \vec{\Phi} \vec{\Phi}^\dag =
\begin{pmatrix} 
\Phi_1^\dag \Phi_1 & \Phi_2^\dag \Phi_1 &\ldots & \Phi_N^\dag \Phi_1 \\
\Phi_1^\dag \Phi_2 & \Phi_2^\dag \Phi_2 &\ldots & \Phi_N^\dag \Phi_2 \\
\vdots & & \ddots & \vdots \\
\Phi_1^\dag \Phi_N & \Phi_2^\dag \Phi_N &\ldots & \Phi_N^\dag \Phi_N
\end{pmatrix}
\end{align}
Moreover, let $\lambda_\alpha$ be generalized Gell-Mann matrices, which is a basis for the real vector space of Hermitian $N\times N$ matrices. 
Then the $N^2$ linearly independent Hermitian bilinears in the fields $\Phi_1,\ldots,\Phi_N$ can be written
\begin{align}\label{E:Kalpha}
	K_\alpha=\text{Tr}(\tilde{K}\lambda_\alpha).
\end{align}
As in \cite{Maniatis:2015gma}, we will let Greek indices like $\alpha$ run from $0$ to $N^2-1$, while Latin indices like $a$
run from $1$ to $N^2-1$. Summation of repeated indices is, as usual, also assumed.  
We will define the matrices $\lambda_\alpha$ such that
the $SO(4)_C$-violating bilinears $\widehat{C}$ are ordered first, that is
\begin{align}\label{E:custViolBilinsInKa}
	K_a = 2\widehat{C}_a = 2\widehat{C}_{m(a),n(a)}, \quad \text{for} \; 1 \leq a \leq \frac{N(N-1)}{2}, 
\end{align}
see appendix \ref{sec:GeneralizedGellMannMatrices}, for a construction of such matrices. We will
put the custodial symmetry-violating bilinears first, to make the conditions for the custodial symmetry simpler to express. 

The very first bilinear $K_0$ will be defined as in \cite{Maniatis:2015gma},
\begin{align}
	K_0=\text{Tr}(\tilde{K}\lambda_0)= \sqrt{\frac{2}{N}}(\Phi_1^\dag \Phi_1+\ldots+\Phi_N^\dag \Phi_N),
\end{align}
with $\lambda_0=\sqrt{2/N}I$.
 Now the general NHDM potential may also be written in the same manner as in \cite{Maniatis:2015gma}
\begin{align}\label{E:potentialAdjoint}
	V=\xi_0 K_0 + \xi_a K_a + \eta_{0} K_0^2 + 2K_0\eta_a K_a+ K_a E_{ab} K_b,
\end{align}
 although the bilinears $K_\alpha$ are organized in a different order here. In \eqref{E:potentialAdjoint}
the parameters $\xi_0, \xi_a, \eta_0, \eta_a$ are real, and $E$ is a real and symmetric $(N^2-1)\times(N^2-1)$ matrix. We will refer to the last term of \eqref{E:potentialAdjoint} as $V_E$,
that is
\begin{align}\label{E:V_Edef}
	V_E= K_a E_{ab} K_b.
\end{align}
Under the basis transformation \eqref{E:HbasisTrafo} the matrix $\tilde{K}$ will transform as
\begin{align}
	\tilde{K} \to \tilde{K}' = U \tilde{K} U^\dag,
\end{align}
while
\begin{align}
	K_0'=K_0, \quad K_a = R_{ab}(U)K_b \quad \text{for}\; a\geq 1,  
\end{align}
since $K_a'= \text{Tr}(U\tilde{K}U^\dag \lambda_a)=\text{Tr}(\tilde{K}U^\dag \lambda_a U)=R_{ab}(U)\text{Tr}(\tilde{K}\lambda_b)=
R_{ab}(U)K_b$,
where the matrix $R_{ab}(U)\in {Ad}_{SU(N)}$ is defined by
\begin{align}\label{E:defR}
	U^\dag \lambda_a U= R_{ab}(U)\lambda_b.
\end{align}
 Here the $\lambda_a$'s may be generalized Gell-Mann matrices, or any basis for the Lie algebra $su(N)$, although we will stick to the generalized Gell-Mann matrices $\lambda_a$ defined in appendix \ref{sec:GeneralizedGellMannMatrices}.
This means the bilinears $K_a$ transform under the adjoint representation of $SU(N)$, while $K_0$
transforms under the trivial representation of $SU(N)$. If the trace $\text{tr}(\lambda_i\lambda_j)\propto \delta_{ij}$, as it will be for generalized Gell-Mann matrices, the matrix $R(U)\in {Ad}_{SU(N)}\subset SO(N^2-1)$ (see appendix \ref{sec:WhenisAd_{SU(N)}subsetSO(N^2-1)}), where the latter inclusion is strict for $N\geq3$.
Eq.~\eqref{E:defR} can be written
\begin{align}\label{E:defR0}
	\lambda_a = R_{ab}(U) U\lambda_b U^\dag,
\end{align}
and by multiplying \eqref{E:defR0} by $(R^{-1})_{ca}$, we obtain
\begin{align}\label{E:defRv20}
	R^{-1}_{ca} \lambda_a = U\lambda_c U^\dag.
\end{align}
Moreover, comparing \eqref{E:defR} and \eqref{E:defRv20} gives us 
$R^{-1}(U)=R(U^\dag)$. If
 $\text{tr}(\lambda_i\lambda_j)\propto \delta_{ij}$, $R$ will be an orthogonal matrix with $R^{-1}=R^T$, and
\eqref{E:defRv20} then yields
\begin{align}\label{E:defRv2}
	R_{ac} \lambda_a = U\lambda_c U^\dag,
\end{align}
which we will apply later.

If we make the substitution $\vec{\Phi} \to U\vec{\Phi}$, and hence
$K_a\to R_{ab}(U)K_b$ in the potential \eqref{E:potentialAdjoint}, the potential remains invariant if we simultaneously substitute the
parameters of the potential with
\begin{align}\label{E:transf2&4-}
	\vec{x}\to  R(U)\vec{x} \quad  \text{and} \quad
	x_0\to  x_0 \quad \text{for} \; x\in\{\xi, \eta\},
\end{align}
 and
\begin{align}
	E\to E' = R(U)ER^T(U).
\end{align}
 
We will now find basis-independent conditions for $SO(4)_C$ symmetry of the different parts of the potential, considered isolated. Later, we will patch these conditions together for sufficient and necessary conditions for custodial symmetry in the general NHDM potential.

\subsubsection{Quadratic terms and quartic terms proportional to $K_0$}
The quadratic terms given by the parameters $\xi_a$ and the quartic terms given by the parameters $\eta_a$, see \eqref{E:potentialAdjoint}, transform under $SU(N)$ Higgs basis transformations in
exactly the same manner, confer \eqref{E:transf2&4-}.
Since the first $N(N-1)/2$ bilinears $K_a$ correspond to custodial symmetry-violating operators of the type $\widehat{C}$, they must be possible to transform away by some Higgs $SU(N)$ basis transformation, when the terms $\xi_aK_a$ and $\eta_a K_a$ are custodially symmetric.
Necessary and sufficient conditions for having $SO(4)_C$ symmetry in these terms simultaneously, is then the existence of a Higgs basis transformation given by $R(U)\in Ad_{SU(N)}$, such that
\ba\label{E:custInK0terms}
   R_{ij}\xi_j = 0 \quad \text{and} \quad  R_{ij}\eta_j = 0 \quad \text{for all} \quad 1 \leq i \leq \frac{N(N-1)}{2}.
\ea

\subsubsection{Quartic terms not involving $K_0$}
If a potential
\ba
   V_E=  K_a E_{ab} K_b
\ea
is custodially symmetric, the real, symmetric matrix $E$ has to be similar to an $E'$ given by
\ba
    E' = R ER^T,
\ea
 where the $N(N-1)/2$ first rows and columns of $E'$ consist of zeros, i.e.
\ba\label{E:EcorrForm}
   E'=\begin{pmatrix} 
 0 & \ldots & 0 &\ldots & 0 \\
\vdots && \vdots & & \vdots \\
0 & \ldots & 0 &\ldots & 0 \\
\vdots & &\vdots & {\bigXdown}& \\
0 & \ldots & 0 & & 
\end{pmatrix},
\ea
\newline
\newline
where $X$ is an arbitrary $(N+2)(N-1)/2\times (N+2)(N-1)/2$ block. The matrix $X$ will be real
and symmetric, since $R$ and $E$ are real, and since $E$ is symmetric.
The first $N(N-1)/2$ rows and columns of $E'$ have to equal zero to make all terms containing custodial symmetry-violating bilinears $\widehat{C}$ disappear. Moreover, this transformation has to be made by some matrix $R\in {Ad}_{SU(N)} \subset SO(N^2-1)$. The first condition that has to be met to make this possible, is that
\ba\label{E:nullityEgeq}
   \text{Nullity}(E)\geq \frac{N(N-1)}{2},
\ea
which means $E$ has an eigenvalue zero with a multiplicity of at least $N(N-1)/2$, or equivalently, the nullspace of $E$ has dimension  $N(N-1)/2$ or more.
We will now prove that the matrix $E$ can be transformed to the form given by \eqref{E:EcorrForm}
only by a matrix $R\in SO(N^2-1)$ which has $N(N-1)/2$ orthonormal nullvectors of $E$ as its first rows.
\begin{proposition}\label{P:NullityOfEAndFormOfR}
Assume the matrix $E$ has at least $k=N(N-1)/2$ linearly independent nullvectors, and let $R\in SO(N^2-1)$ be such that $E'=RER^T$. Then the 
$k$ first rows and columns of $E'$ are zero, as given in \eqref{E:EcorrForm}, if and only if $R$ is of the form
\ba\label{E:Rform}
   R= \left[ n_1, \ldots, n_k, c_{k+1}, \ldots, c_{N^2-1} \right]^T,
\ea
where $n_1, \ldots, n_k$ are nullvectors of $E$ and $n_1, \ldots, n_k, c_{k+1}, \ldots, c_{N^2-1}$ are orthonormal column vectors. 
\end{proposition}
\begin{proof}
 ($\Leftarrow$): Assume $R$ is of the form given by \eqref{E:Rform}. Since $n_1, \ldots, n_k$ are nullvectors of $E$, the $k$ first columns of the product $(E R^T)$ are zero, and hence also the $k$ first columns of $E' =R E R^T$ are zero. Now $E'$ is symmetric since $E$ is symmetric, and consequently the $k$ first rows of $E'$ are also zero.

($\Rightarrow$): Assume the $k$ first columns and rows of $E'=RER^T$ are zero. 
Let $R^T=\left[c_1, \ldots,c_k, c_{k+1}, \ldots c_{N^2-1}\right]$. Since $R^TR =I$, the column vectors $\{c_j\}$ are orthonormal. 
Write $E=\left[e_1,\ldots,e_{N^2-1}\right]^T$. Let $(h_l)_j=e_j^Tc_l$, where $l\leq k$. Then $h_{l}$ is the $l$'th column in the product $(ER^T)$ for a fixed $l\leq k$. We will now show that $(h_{l})_j=0$ for all $j$ and $l\leq k$, which means that $c_l$ is a nullvector of $E$: $E_{il}'=(c_{i})_j (h_l)_j=0$ by assumption. This means 
that $h_l \bot c_i$ for all $i$. But the set $\{c_i \}_{i=1}^{N^2-1}$ is linearly independent, since $R$ was invertible. This infers that
$\{c_i \}_{i=1}^{N^2-1}$ spans all $\N{R}^{N^2-1}$, and hence $h_l \bot c_i$ for all $i$ cannot be true unless $h_l=0$. Hence the components $(h_l)_j=0$ for all $j$, and then $c_l$ is a nullvector of $E$ for all $l\leq k$.
\end{proof}
The choice, or permutations, of nullvectors $n_1, \ldots, n_k$, does not affect $E'$, since the vanishing elements of $E'$
are the only ones to involve these nullvectors: $E'=RER^T$ infers 
\ba
E_{ij}'= R_{i\alpha } E_{\alpha \beta} R_{ j\beta}= (c_i)_\alpha E_{\alpha \beta} (c_j)_\beta=0
\ea
 when $i$ or $ j \leq k =N(N-1)/2$, and where $\{c_l \}_{l=1}^{N^2-1}=\{n_1, \ldots, n_k, c_{k+1}, \ldots, c_{N^2-1}\} $. 

We will now derive sufficient and necessary conditions for when a matrix $R$, for instance of the form \eqref{E:Rform}, is a member of the adjoint representation of $SU(N)$,
that is $R\in Ad_{SU(N)} \subset SO(N^2-1)$.
To obtain this, we will need some results about Lie algebras. 
Let $g$ be a Lie algebra. By a Lie algebra automorphism $r$ we will mean a $\N{R}$-linear 
bijection on the vector space $g$, such that the Lie bracket is preserved. That is, 
$r$ is a injection (1-1) from $g$ onto $g$, and if $X=x_i \lambda_i \in g$,
then $r(X)=x_i r(\lambda_i)$, and
\begin{align}
	r([X,Y])=[r(X),r(Y)] 
\end{align}
for all $X,Y\in g$. When $g=su(N)$, automorphisms are either similarity transformations (inner automorphisms) or combinations of similarity transformations and complex conjugation (outer automorphisms):
\begin{proposition}\label{T:auto}  
An automorphism $r: su(N) \to su(N)$ for $N>2$ is either an inner automorphism, i.e.~a similarity transformation
\[
   r(X)= U X U^\dag,
\]
for an $U\in SU(N)$, or a combination of complex conjugation and a similarity transformation,
\ba\label{E:outerForm}
   r(X)=U X^\ast U^\dag.
\ea
\end{proposition}
\begin{proof}
Similarity transformations $r(X)= U X U^\dag$ are Lie algebra automorphisms since they are the derivatives of Lie group automorphisms $\psi(V)=UVU^\dag$ for $U,V\in SU(N)$, see e.g.~\cite{Baker}. These are the inner automorphisms. All non-inner automorphisms are called outer automorphisms.
For $su(N)$ with $N>2$, the outer automorphisms consist of complex conjugation, in combination with an inner automorphism: 

The outer isomorphism group $\text{Out}(\mathfrak{g})$
of the a real, simple Lie algebra $\mathfrak{g}$, is always given by $\text{Out}(\mathfrak{g})=\text{Aut}(\mathfrak{g})/\text{Inn}(\mathfrak{g})$, just as is the case for complex, simple Lie algebras \cite{gundogan}. The real, simple Lie algebra $su(N)$ has an outer automorphism group
 $\text{Out}(su(N))$ isomorphic to the outer automorphism group of the complexification of real
$su(N)$, that is $su(N,\N{C})=sl(N,\N{C})$. The outer automorphism  group $\text{Out}(su(N))$ is hence the automorphism group of the Dynkin diagram $A_{N-1}$, which is trivial for $N=2$, and isomorphic to $\N{Z}_2$ for $N>2$, where the non-trivial element of $\text{Out}(su(N))$ in the latter case corresponds to complex conjugation. This is a well known result, but see e.g.~\cite{chuah} with the compact $su(N)=\mathfrak{g}$  as a real form of the complex, simple Lie algebra $sl(N,\N{C})$, where a Cartan involution $\theta$ of $su(N)$ only generate the trivial group, and hence $\text{Out}(su(N))\cong \text{Aut}(A_{N-1})\cong \text{Out}(sl(N,\N{C}))$. 


Finally, since $(UXU^\dag)^\ast = U^\ast X^\ast U^{\ast \dag}$, and $U^\ast \in SU(N)$ when $U\in SU(N)$,
an outer automorphism can always be written on the form \eqref{E:outerForm}.    
\end{proof}
See appendix \ref{sec:CCisAnOuterAutomorphismOfsu(N)} for an explanation why complex conjugation is an
outer automorphism of $su(N)$ for $N>2$, while it is an inner automorphism of $su(2)$.

We can now show the following characterization of $Ad_{SU(N)}$ for $N>2$: A matrix $R$ is an element of $Ad_{SU(N)}$ if and only if the linear mapping $r$ on $su(N)$ associated with $R$ preserves
the commutator for all elements of $su(N)$, and does not involve complex conjugation. 
\begin{proposition}\label{T:AdInnA} 
  Let $r: su(N) \to su(N)$ be a mapping on the Lie algebra $su(N)$, with $N>2$. Moreover, let 
	$\{v_i\}_{i=1}^{N^2-1}$ be a basis for $su(N)$ with $\text{tr}(v_i v_j)\propto \delta_{ij}$ such that $Ad_{SU(N)}\subset SO(N^2-1)$, and	let $R$ be a real $(N^2-1)\times (N^2-1)$ matrix such that for $X=x_i v_i \in su(N)$, we have 
	\[
	r(X)=(Rx)_i v_i =R_{ij}x_j v_i. 
	\]
	Then $R\in Ad_{SU(N)}$ if and only if the mapping $r$ is an inner Lie algebra automorphism, that is, $r$ is a $\N{R}$-linear bijection and
\[
	r([X,Y])=[r(X),r(Y)] 
\]	
for all $X,Y\in su(N)$, while for all $U\in SU(N)$, $r(X)\ne U X^\ast U^\dag$ for some $X\in su(N)$.
\end{proposition}
\begin{proof}
($\Rightarrow$): Assume $R\in Ad_{SU(N)}$. This means there is an $U\in SU(N)$ such that $R_{ij}v_i=Uv_j U^\dag$ by \eqref{E:defRv2}.
Then $r(Z)=r(z_i v_i)= R_{ij} z_j v_i=z_j U v_j U^\dag= U Z U^\dag$ for any $Z\in su(N)$, and hence $r$ is an inner automorphism. It respects the commutator in the following manner: 
$[r(X),r(Y)]= [ U X U^\dag, U Y U^\dag]=U (X Y  -  Y X) U^\dag = r([X,Y])$. By proposition \ref{T:auto} an inner automorphism means an automorphism that does not involve complex conjugation.

($\Leftarrow$): Assume $r$ is an inner Lie algebra automorphism on $su(N)$. Then, by definition of an inner automorphism,
$r(X)= U X U^\dag$,
for an $U\in SU(N)$.

 Then $U X U^\dag= r(X)=r(x_i v_i)=x_i r( v_i)=x_j R_{ij} v_i$, and $R\in Ad_{SU(N)}$ by \eqref{E:defRv2}. 
\end{proof}
We will now find conditions on the matrix $R$ equivalent with the associated linear mapping $r$
preserving the commutator.
Let $X, Y\in su(N)$, with $X=x_i i\lambda_i, Y=y_j i\lambda_j$, where $\{v_i\}_{i=1}^{N^2-1} =\{i\lambda_i\}_{i=1}^{N^2-1}$ is a basis for the Lie algebra $su(N)$ with $\text{tr}(v_i v_j)\propto \delta_{ij}$, and where the matrices $\lambda_i$ are satisfying
\ba
   \left[ \lambda_i, \lambda_j\right]=2if^{ijk}\lambda_k.
\ea
The constants $f^{ijk}$ are denoted structure constants. The matrices $\lambda_i$ may be generalized Gell-Mann matrices, or any other matrices such that $\{i\lambda_i\}_{i=1}^{N^2-1}$ is a basis for $su(N)$.
Then $[X,Y]=[x_i i\lambda_i,y_j i\lambda_j]=-x_iy_j[\lambda_i,\lambda_j]=
-x_iy_j(2if^{ijk}\lambda_k)$, hence
\ba\label{E:LScomm}
   r(\left[X,Y\right])=-2 x_iy_jf^{ijk}r(i\lambda_k)=-2x_iy_jf^{ijk}R_{ek}i\lambda_e.
\ea
Note that $r$ is a $\N{R}$-linear function on $su(N)$, with linear combinations over $\N{R}$ of the basis vectors $\{i\lambda_i\}_{i=1}^{N^2-1}$ as domain, and not linear combinations over $\N{R}$ of the matrices $\{\lambda_i\}_{i=1}^{N^2-1}$ as domain.
  
On the other hand, $[r(X),r(Y)]=[x_iR_{ai}i\lambda_a,y_jR_{cj}i\lambda_c]=  x_iR_{ai}y_jR_{cj}[i\lambda_a,i\lambda_c]= x_iR_{ai}y_jR_{cj}(-2if^{ace}\lambda_e)$, and hence
\ba\label{E:RScomm}
    \left[r(X),r(Y)\right]= -2i x_iy_jR_{ai}R_{cj}f^{ace}\lambda_e.
\ea
Equating \eqref{E:LScomm} and \eqref{E:RScomm}, gives us
\ba
  x_iy_jR_{ek}f^{ijk}= x_iy_jR_{ai}R_{cj}f^{ace},
\ea
and since $x_i, y_j$ are arbitrary, we get
\ba\label{E:condComm}
   R_{ek} f^{ijk}= R_{ai}R_{cj}f^{ace}.
\ea
The equations \eqref{E:condComm} being satisfied is equivalent with the mapping $r$ respecting the commutator. If we furthermore assume that $R$ is a bijection, 
it will be invertible, and the inverse will be the transposed matrix $R^T$, since by proposition \ref{T:AdInnA} either $R\in Ad_{SU(N)} \subset SO(N^2-1)$ 
or $R$ is a product of complex conjugation and matrices in $Ad_{SU(N)}$. The latter means, in case the $\lambda_i$'s are generalized Gell-Mann matrices, that $R\in O(N^2-1)$.\footnote{If the basis matrices $\lambda_i$ are either real or purely imaginary, then complex conjugation $r(X)=x_i r(\lambda_i) = x_i \lambda_i^\ast$, which makes complex conjugation $R$ a diagonal matrix with $R_{ii}=\pm 1$ (no sum over $i$), negative if $\lambda_i$ is imaginary. Hence $R\in O(N^2-1)$.}  
In fact, either $R$ will be bijective or it will be zero:
\begin{proposition}\label{t:SimpleLie}
 A Lie algebra homomorphism $r:su(N)\to su(N)$ is either an automorphism (and hence bijective), or $r=0$. 
\end{proposition}
\begin{proof}
  The kernel of a Lie algebra homomorphism $\phi: \mathfrak{g}\to \mathfrak{g'} $ is an ideal of the Lie algebra $\mathfrak{g}$. Furthermore, every Lie algebra ideal 
  $\mathfrak{i}$ corresponds precisely to a homomorphism with kernel $\mathfrak{i}$. A Lie algebra is called simple if every proper ideal $\mathfrak{i}\lhd \mathfrak{g}$ is trivial, that is $\mathfrak{i}=0$. 
  Now $su(N)$ is a simple Lie algebra, which means that only trivial, proper ideals exist. Hereby a homomorphism which is not an automorphism equals zero.
\end{proof}
If we now assume $R\ne 0$, and multiply \eqref{E:condComm} by the inverses of the matrices $R_{ai}$ and $R_{cj}$, that is $R^T_{id}=R_{di}$ and $R^T_{jg}=R_{gj}$, we get\footnote{An alternative characterization of $Ad_{SU(N)}$ is that $Ad_{SU(N)}$ is the matrices $R$ which leave the trace $\text{Tr}(XYZ-XZY)$ invariant for arbitrary $X,Y,Z\in su(N)$. Here $\{X,Y,Z\}$ are simultaneously transformed as $W=w_i\lambda_i\to (Rw)_i\lambda_i= R_{ij}w_j\lambda_i $, for all $W\in\{X,Y,Z\}$.  Applying this characterization, and $\text{Tr}(\lambda_a \lambda_b)\propto \delta_{ab}$, leads to \eqref{E:condComm2} as well.}
\ba\label{E:condComm2}
   R_{di}R_{gj}R_{ek} f^{ijk}= f^{dge}.
\ea

Unfortunately, whether the matrix $R$ of \eqref{E:Rform} is an element of $Ad_{SU(N)}$ or not, will depend on the choice of orthonormal nullvectors $n_1,\ldots, n_k$, where $k=N(N-1)/2$. Let 
\ba
l=\text{Nullity}(E)\geq k, 
\ea
and let $S=\{\tilde{n}_i\}_{i=1}^l$ be an orthonormal set of nullvectors of $E$. Then $S$ spans the nullspace of $E$, and any choice of orthonormal nullvectors 
$\{{n}_i\}_{i=1}^k$ of $R$, can be written as a rotation of the vectors of $S$. This means that any linear combination of the nullvectors of $S$ can be written as
\ba\label{E:Oltrafo}
   \left[n_1,n_2,\ldots,n_k,n_{k+1},\ldots, n_l \right]^T= O \left[\tilde{n}_1,\tilde{n}_2,\ldots, \tilde{n}_l \right]^T,
\ea
where the matrices of nullvectors are regarded as $1\times l$ matrices with the nullvectors as elements,
so that the transpose does not act on the nullvectors.
The superfluous vectors $n_{k+1},\ldots, n_l$ will not be applied in the construction of $R$, and the $l\times l$ matrix $O$ has to be
an element of the orthogonal group $O(l)$ to ensure the vectors $n_{1},\ldots, n_l$ are orthonormal: $n_i\cdot n_j = \delta_{ij}$ means that
$O_{ip}\tilde{n}_p \cdot O_{jq}\tilde{n}_q = O_{ip} O_{jq}\delta_{pq} =   \delta_{ij}$,
  which infer $O_{ip}O_{jp}= O_{ip}O_{pj}^T =\delta_{ij}$, that is, $OO^T=I$, and we conclude that
\ba
   O\in O(l).
\ea
 
 Matrices $O$ that change the orders of two nullvectors are examples of elements in the orthogonal group $O(l)$. To see that the order of the nullvectors matters, take for instance $SU(3)$ with the matrices given in appendix \ref{sec:GeneralizedGellMannMatrices} as $su(3)$-basis. Let $R$ be defined by \eqref{E:Rform}, that is,
\ba\label{E:Rgjentatt}
R=[n_1,n_2, n_3, \ldots,n_k,c_{k+1}, \ldots, c_{N^2-1}]^T, 
\ea
and let 
\ba\label{E:R'}
R'=[n_2,n_1,n_3,\ldots,n_k,c_{k+1}', \ldots, c_{N^2-1}']^T, 
\ea
i.e.~with identical nullvectors as $R$ except for the two first rows being interchanged. Assume $R$ satisfies \eqref{E:condComm2} for some choice of $c_{k+1},\ldots,c_{N^2-1}$. We will show that $R'$ can not satisfy
\eqref{E:condComm2} for any choice of $c_{k+1}',\ldots,c_{N^2-1}'$, and hence the order of the nullvectors matters when we want to test for the custodial symmetry:
All the equations given by different values for $d,g$ and $e$ in \eqref{E:condComm2}
have to be satisfied simultaneously, if the quartic terms shall be custodially symmetric.
Since $R$ satisfies \eqref{E:condComm2} for all choices of $d,g$ and $e$, we must for instance have\footnote{The structure constant $f^{123}=1/2$ with our alternatively ordered Gell-Mann matrices, defined in appendix \ref{sec:GeneralizedGellMannMatrices}, while $f^{123}=1$ with the standard Gell-Mann matrices.}
\begin{align}\label{E:contrR}
	\frac{1}{2} = f^{123}= R_{1i}R_{2j}R_{3k} f^{ijk}.
\end{align}
On the other hand, if $R'$ shall be an element of $Ad_{SU(3)}$, we must for instance have
\begin{align}\label{E:contrR'}
	-\frac{1}{2}= f^{213}= R_{2i}'R_{1j}'R_{3k}' f^{ijk},
\end{align}
where $R_{2i}'$ is the second row of $R'$, that is $R_{2i}'=(n_1)_i$ by \eqref{E:R'},
while $R_{1j}'=(n_2)_j$ and $R_{3k}'=(n_3)_k$ ($N(N-1)/2=3$ when $N=3$, and hence we must have at least 3 nullvectors in case of a custodial symmetry).
Then \eqref{E:contrR'} yields
\begin{align}\label{E:contrad}
	-\frac{1}{2} = R_{2i}'R_{1j}'R_{3k}' f^{ijk} = R_{1i}R_{2j}R_{3k} f^{ijk},
\end{align}
which contradicts \eqref{E:contrR}, and $R$ and $R'$ cannot simultaneously be elements of $Ad_{SU(N)}$.
Moreover, if the nullity of $E$ is greater than $3$, and if $n_1$ in \eqref{E:Rgjentatt} is interchanged with an entirely new nullvector $\bar{n}_1$, eq.~\eqref{E:contrR} with $R_{1i}=(\bar{n}_1)_i$ may not hold anymore. And if \eqref{E:contrR} still holds, it will not hold anymore if $\bar{n}_1$ is interchanged with $n_2$, as in the previous example.

\subsubsection{Necessary conditions for the custodial $SO(4)$ symmetry}
\label{sec:Necessary_conditions_for_the_custodial_symmetry}
The equations \eqref{E:condComm2} give us a simple test, that is, a necessary 
condition, for the $SO(4)_C$ symmetry. Choosing $d,g,e\leq N(N-1)/2$, reduces all elements from the matrix $R$ in \eqref{E:Rform} to elements of the chosen nullvectors,
\ba\label{E:condComm2Test}
   f^{dge}= (n_{d})_i (n_{g})_j (n_{e})_k f^{ijk}.
\ea
If there for any choice of the nullvectors in
\eqref{E:condComm2Test} is a choice of $d,g,e$ (dependent on the choice of nullvectors), such that \eqref{E:condComm2Test} does not hold, then the quartic terms $V_E$ (and the potential) are not $SO(4)_C$-symmetric. 
The advantage with \eqref{E:condComm2Test}, is that it only refers to the nullvectors of $E$ and the structure constants of $su(N)$, which are simple to calculate.

Consider a specific set of $l=\text{Nullity}(E)$ orthonormal nullvectors
$S=\{\tilde{n}_i\}_{i=1}^{i=l}$, they can be rotated into any orthonormal set of nullvectors $n_1,n_2,\ldots,n_k,n_{k+1},\ldots, n_l$ of $E$, as given by \eqref{E:Oltrafo}.
Then the necessary condition \eqref{E:condComm2Test} for the custodial symmetry expressed with specific, orthonormal nullvectors
$\{\tilde{n}_i\}_{i=1}^{i=l}$ becomes
\ba\label{E:condComm2TestO}
     f^{dge}= (O_{dp}\tilde{n}_p)_i (O_{gq}\tilde{n}_q)_j (O_{er}\tilde{n}_r)_k f^{ijk}. 
\ea
If there for all choices of $O\in O(l)$ exist
choices of indices $d,g,e\leq k=N(N-1)/2$ such that
\eqref{E:condComm2TestO} does not hold, then
the quartic terms $V_E$ (and hence the whole potential) are not $SO(4)_C$-symmetric.

For the 3HDM \eqref{E:condComm2TestO} can be substantially simplified, in the case of three nullvectors, i.e.~$l=3$.  
For the 3HDM all equations \eqref{E:condComm2TestO} will hold automatically, perhaps except for the case where $(dge)=(123)$: If
 two of the indices $d,g,e$ are equal, the left hand side of \eqref{E:condComm2TestO} is zero, and the right hand side is also zero,
since the expression will be odd in two of the indices. For instance, if $d=g=1$, then the two first factors of the right hand side,
$(O_{1p}\tilde{n}_p)_i$ and $(O_{1q}\tilde{n}_q)_j$, are the same, and since $f^{ijk}$ is antisymmetric in $i$ and $j$, the sum over $i$ and $j$ will be zero for each $k$. Now write $f^{ijk}=\epsilon^{ijk} f(i,j,k)$, where $\epsilon^{ijk}$ is the completely antisymmetric Levi-Civita symbol. Then
\ba\label{E:condComm2TestO_Simplified3HDM}
    f^{123} &&= \epsilon^{ijk} O_{1p}(\tilde{n}_p)_i O_{2q}(\tilde{n}_q)_j O_{3r}(\tilde{n}_r)_k f(i,j,k) \nn \\
		  &&= \epsilon^{ijk} \epsilon^{pqr} O_{1p}O_{2q}O_{3r}(\tilde{n}_1)_i (\tilde{n}_2)_j (\tilde{n}_3)_k f(i,j,k) \nn \\
			&&= \text{det}(O) (\tilde{n}_1)_i (\tilde{n}_2)_j(\tilde{n}_3)_k f^{ijk},
\ea
since $\epsilon^{ijk} A_{ip}A_{jq}A_{kr}=\epsilon^{ijk} \epsilon^{pqr} A_{i1}A_{j2}A_{k3}$ for any square matrix $A$, with $A_{ab}=(\tilde{n}_b)_a$ in our case. We have also used that $\text{det}(O)=\epsilon^{pqr} O_{1p}O_{2q}O_{3r}$. Hence, for the 3HDM in the case $E$ has exactly $3$ nullvectors, the necessary condition \eqref{E:condComm2TestO} holds for all $SO(3)$-rotated nullvectors, if and only if it holds for the initial nullvectors $\tilde{n}_1$, $\tilde{n}_2$ and $\tilde{n}_3$.  
 We will apply \eqref{E:condComm2TestO_Simplified3HDM} to show that certain potentials are not custodially symmetric in section \ref{sec:ExamplesApplyingANecessaryCondition}.

\subsubsection{Necessary and sufficient conditions for the custodial $SO(4)$ symmetry}
\label{sec:Necessary_and_sufficient_conditions}
We will now summarize our results from the previous sections as a theorem, giving necessary and sufficient, basis-independent conditions for having the custodial $SO(4)$ symmetry in a NHDM.
Let $V(\Phi_1,\ldots,\Phi_N)$ be a NHDM potential, given by \eqref{E:potentialAdjoint}. If there exists a Higgs basis transformation $\vec{\Phi}\to \vec{\Phi}'=U\vec{\Phi}$, such that $V'(\Phi_1',\ldots,\Phi_N')=V(\Phi_1,\ldots,\Phi_N)$, and $V'(\Phi_1',\ldots,\Phi_N')$ is manifestly $SO(4)_C$-symmetric, we called the original potential $V(\Phi_1,\ldots,\Phi_N)$ $SO(4)_C$-symmetric, cf.~the discussion below \eqref{E:HbasisTrafo}. We can then show the following:
\begin{theorem}\label{T:tm}
 Let $V$ be a NHDM potential, given by \eqref{E:potentialAdjoint}, with $N\geq 3$. Then the potential $V$ is $SO(4)_C$-symmetric if and only if the following three conditions are satisfied simultaneously:
\newline
i)
 The nullity 
 $l$ of the matrix $E$ of \eqref{E:potentialAdjoint} is equal to or greater than
$k=N(N-1)/2$.
\newline
ii) There exists a real $(N^2-1)\times (N^2-1)$ matrix $R$ whose $N(N-1)/2$ first rows are an orthonormal set of nullvectors of $E$, such that 
\ba\label{E:condComm3}
    f^{abc}=R_{ai}R_{bj}R_{ck}f^{ijk},
\ea
is satisfied for all $a,b$ and $c$. The constants $f^{ijk}$ here are the structure constants associated with the alternatively ordered, generalized Gell-Mann matrices $\{ \lambda_j\}_{j=1}^{N^2-1}$ of appendix \ref{sec:GeneralizedGellMannMatrices}. 
\newline
iii)  The matrix $R$ of condition ii) also satisfies
\ba\label{E:custInK0terms2}
    R_{ij}\xi_j = 0 \quad \text{and} \quad  R_{ij}\eta_j = 0 \quad \text{for all} \quad 1 \leq i \leq \frac{N(N-1)}{2},
\ea 
where $\xi_j$ and $\eta_j$ are given by \eqref{E:potentialAdjoint}.
\newline

Moreover, the solutions of \eqref{E:condComm3} will come in pairs $R_1$ and $R_2=R^{\text{cc}}R_1$, where 
\ba
R^{\text{cc}}= \begin{pmatrix}
I_{k\times k} & 0 \\
0 & -I_{m\times m }
\end{pmatrix}, 
\ea
with $m=N^2-1-k$, and $k$ given above. The matrix
$R^{\text{cc}}$ represent complex conjugation when acting on the Lie algebra $su(N)$ with the basis
$\{i \lambda_j\}_{j=1}^{N^2-1}$ given in appendix \ref{sec:GeneralizedGellMannMatrices}. Exactly one of the solutions $R_1$ and $R_2$ will correspond to a $SU(N)$ basis transformation of the Higgs fields.

Finally, condition iii) will hold if the column vectors $\vec{\xi}$ and $\vec{\eta}$ are orthogonal to the nullspace of $E$.
\end{theorem}
\begin{proof}
As stated in \eqref{E:nullityEgeq}, the matrix $E$ of the quartic terms $V_E$ must have $\text{Nullity}(E)=l\geq k=N(N-1)/2$, and the $k$ first rows of $R$ must by 
proposition \ref{P:NullityOfEAndFormOfR}
be orthonormal nullvectors of $E$, for $E$ to be transformed by $R$ to the form $E'$ given by \eqref{E:EcorrForm}. 
The matrix $R$ must be a bijection, satisfy \eqref{E:condComm3} and not involve complex conjugation on the Lie algebra $su(N)$ with basis $\{i\lambda_j\}_{j=1}^{N^2-1}$ for $R$ to be an element of ${Ad}_{SU(N)}$, by \eqref{E:condComm2} and proposition \ref{T:AdInnA}. By proposition \ref{t:SimpleLie}, $R$ is bijective if it does not equal zero, and $R\ne 0$ since the $k$ first rows of $R$ are nullvectors of $E$.
If $R_1$ is a solution of  \eqref{E:condComm3},  $R_2=R^{\text{cc}}R_1$ will also be a solution, since $R^{\text{cc}}$ represents complex conjugation which is an automorphism on $su(N)$, and since $R_2$ has the same $k$ first rows as $R_1$, namely the chosen nullvectors of $R$.
Exactly one of these two solutions are elements of $Ad_{SU(N)}$, since it can be written as an inner automorphism of $su(N)$, the other will be an outer automorphism: If $R$ is an inner automorphism, we can write $r(X)=(R_{kj}x_j)(i\lambda_k)=U XU^\dag$, with $X=x_j(i\lambda_j)$, and then  $(R^{\text{cc}}_{lk} R_{kj}x_j)(i\lambda_l)= (UXU^\dag)^\ast$ for some $U\in SU(N)$,
is an outer automorphism. On the other hand, if $R$ is outer, then $R=R^{\text{cc}} R$ will be inner since two complex conjugations will give an ordinary similarity transformation of $X$.

By \eqref{E:custInK0terms} the necessary and sufficient conditions for having the custodial $SO(4)$ symmetry in the terms additional to $V_E$ in the potential is given by \eqref{E:custInK0terms2}. Eq.~\eqref{E:custInK0terms2} will hold if $\vec{\xi}$ and $\vec{\eta}$ are orthogonal to the nullspace of $E$,
since the $k$ first rows of $R$ are orthonormal nullvectors of $E$. 
If $l=k$,
then \eqref{E:custInK0terms2} will hold if and only if $\vec{\xi}$ and $\vec{\eta}$ are orthogonal to the nullspace of $E$, since
the $k$ first rows of $E$ then spans the whole nullspace of $E$. 
\end{proof}
 The $N(N-1)/2$ first columns of $R$ had to consist of orthonormal nullvectors of $E$, for $R$ to be of the form \eqref{E:Rform} necessary to transform the matrix $E$ into a manifestly $SO(4)_C$-symmetric matrix $E'$. When we are searching for a matrix $R$ as described in the theorem, $l$ concrete, orthonormal nullvectors of $E$ can be rotated by $O(l)$ transformations to find a matrix $R$ which satisfies the conditions of theorem \ref{T:tm}. This was discussed in connection with \eqref{E:Oltrafo}.

The system of equations given by \eqref{E:condComm3} is overdetermined. Permutations of the indices $a,b$ and $c$ give equivalent
equations, since the structure constants are totally antisymmetric in all indices for $su(N)$, cf.~e.g.~\cite{metha}. Moreover, if (at least) two indices 
among $a,b$ and $c$ are identical, the left hand side of \eqref{E:condComm3} will be zero, and the right hand side will be zero too: If, for instance $a=b$ the expression on the right hand changes sign if we interchange the indices $i$ and $j$, and hence the sums over $i$ and $j$ equal zero for each $k$. Hence, there are 
\ba\label{E:NumberEqs}
\binom{N^2-1}{3}=  \frac{1}{6} \left(N^2-3\right) \left(N^2-2\right) \left(N^2-1\right)
\ea
equations left, which equals $56$ equations for the 3HDM.
Moreover, there are $(N^2-1)^2$ elements in $R$, and if we subtract the $(N^2-1)\cdot N(N-1)/2$ elements associated with nullvectors of $E$, we end up with  $({1}/{2})(N-1)^2 (N+1) (N+2)$ variables, which equals $40$ for $N=3$. If we add the $\binom{k}{2}=({1}/{8}) (N-2) (N-1) N (N+1)$, 
where $k=N(N-1)/2$, variables $SO(k)$ rotations of the nullvectors generate (we here assume the nullity $l$ of $E$ is $k$), we get  totally $({1}/{8}) (N-1) (N+1) (N (5 N+2)-8)$ variables. In the 3HDM this corresponds to $43$ variables. The difference between the number of equations and variables is then 
\ba
\frac{1}{24} (N-1) (N+1) \left(N \left(4 N^3-35 N-6\right)+48\right)>0, 
\ea
for $N\geq 3$, and the system of equations \eqref{E:condComm3} is hence overdetermined.

In condition ii) in theorem \ref{T:tm}, we state that $R$ should not represent complex conjugation of elements of $su(N)$ expressed by the basis $B=\{i \lambda_j\}_{j=1}^{N^2-1}$ of appendix \ref{sec:GeneralizedGellMannMatrices}. Since the $N(N-1)/2$ first elements of the basis $B$ are real, while the other elements are purely imaginary, this meant that complex conjugation is represented by a matrix $R$ with $1$'s on the $N(N-1)/2$ first diagonal elements, and $-1$ on all the other diagonal elements, while $R$ is zero elsewhere. This matrix $R$ will satisfy \eqref{E:condComm3}, but is not an element of ${Ad}_{SU(N)}$, and does hence not correspond to a Higgs basis transformation. For the 3HDM, $R$ representing complex conjugation then reads
\ba\label{E:Rcc}
 R^{\text{cc}}= \begin{pmatrix} 
I_{3\times 3} & 0_{3\times 5} \\
0_{5\times 3} & -I_{5\times 5}
\end{pmatrix},
\ea
where $I_{n\times n}$ are $n\times n$ identity matrices and $0_{m\times n}$ are $m\times n$ zero matrices.

In the case $N=2$, theorem \ref{T:tm} also applies, although condition i) will imply condition ii) in this case. The reason for this is that $Ad_{SU(2)}=SO(3)$, and hence any $R \in SO(3)$ will satisfy 
\eqref{E:condComm3}, which will hold for all $R\in Ad_{SU(N)}$ also in the case $N=2$. And given a normalized nullvector $n_1$ of $E$, you can always construct a matrix $R\in SO(3)$ with $n_1$ as the first row. Then, since $R\in SO(3)=Ad_{SU(2)}$, \eqref{E:condComm3} will hold.
The fact that \eqref{E:condComm3} holds for all $R\in SO(3)$ can also be shown by an explicit calculation: For $N=2$, the structure constants
\ba
f^{abc}\propto \epsilon^{abc},
\ea
where $\epsilon^{abc}$ is the Levi-Civita symbol, and hence \eqref{E:condComm3} in this case becomes equivalent to
\ba\label{E:Ad_SU(2)cond}
\epsilon^{abc}=R_{ai}R_{bj}R_{ck}\epsilon^{ijk}.
\ea
But, similarly to the argument below \eqref{E:condComm2TestO_Simplified3HDM},
$R_{ai}R_{bj}R_{ck}\epsilon^{ijk}=R_{1i}R_{2j}R_{3k}\epsilon^{ijk}\epsilon^{abc} =
\det(R) \epsilon^{abc}=\epsilon^{abc}$, where the last equality is valid when $\det(R)=1$,
and hence \eqref{E:Ad_SU(2)cond} holds for all $R\in SO(3)$.

 On the other hand, condition ii) of theorem \ref{T:tm} implies condition i) for any $N$, and hence these two conditions are equivalent for $N=2$. (Condition i) is in any case just a first, simple necessary condition for $SO(4)_C$ symmetry.) 
But, for $N=2$ as for all $N$, the matrices $R$ which fulfil condition ii) may not satisfy condition iii). Hence, if there is a matrix $R$ which satisfy
both condition ii) and iii), then the potential is $SO(4)_C$-symmetric also for $N=2$, although condition ii) is more trivial in this case. These conditions now render the conditions given in \cite{Grzadkowski:2010dj}.

Finally, for $N=2$ both solutions $R_1$ and $R_2=R^{\text{cc}}R_1$ mentioned in theorem \ref{T:tm}, will correspond to a $SU(2)$ Higgs basis transformation, since complex conjugation is an inner automorphism of $SU(2)$.

\subsubsection{Spontaneous breaking of the custodial $SO(4)$ symmetry}
\label{sec:SpontaneousBreakingCustodialSymmetry}
We have referred to a NHDM potential with a custodial $SO(4)$ symmetry as $SO(4)_C$-symmetric.
In the presence of VEVs $v_n$ given by
\ba\label{E:VEVnhdm}
    \langle 0| \Phi_n |0\rangle = (0, v_n)^T,
\ea
the $SO(4)_C$ symmetry is broken or, more precisely, hidden. The complex VEVs $v_n$ here occurs in the lower elements of the doublets, due to electrical charge conservation. 
Given at least two non-zero VEVs in a manifestly $SO(4)_C$-symmetric potential, the least amount of symmetry breaking occurs in the case of vacuum alignment, i.e.~when all VEVs 
occur in the
same direction of the real quadruplets $\Phi_{n,r}$, cf.~\eqref{E:realQuadrNHDM}. 
This means that all VEVs can be written 
\ba
v_n= \tilde{v}_n e^{i\theta}, 
\ea
where $\tilde{v}_n$ and $\theta$ are real, and we can transform all VEVs to real VEVs without altering the parameters of the potential, by making an $U(1)$ phase transformation on all scalar fields simultaneously. 
Normally only $U\in SU(N)$ are considered as Higgs basis transformations since an overall $U(1)$ transformation does not affect the parameters of the potential. 
  But when we allow for $U\in U(N)$, i.e.~allow overall complex phases in addition to the $SU(N)$ Higgs basis transformations, vacuum alignment is equivalent to all VEVs being real.
Then, in case of real VEVs, $SO(4)_C$ symmetry is spontaneously broken down to $SO(3)_C$, with three broken generators \cite{Olaussen:2010aq}. 

If there is a Higgs basis where the potential is manifestly $SO(4)_C$-symmetric, where
$SO(3)_C$ at the same time is intact, then the potential is simultaneously
explicitly and ``spontaneously'' custodially symmetric, i.e.~custodially symmetric. This will happen if and only if there exists a matrix $R$ which satisfies condition ii) and iii) of theorem \ref{T:tm}, where the vacua 
 \ba
 (v_1, \ldots, v_N)^T=\langle 0| U(R)\vec{\Phi}|0 \rangle
 \ea
  are real (i.e.~aligned), for a Higgs basis transformation $U(R)\in U(N)$ associated with $R$ through \eqref{E:defR}. (There are $N$ matrices $U$ in $SU(N)$ associated with each $R$, and when $U(R)\in U(N)$ an additional complex phase may be present, without affecting the matrix $R$).  Condition i) of theorem \ref{T:tm} is, as we commented in the end of section \ref{sec:Necessary_and_sufficient_conditions}, a consequence of condition
 ii), and will hence be satisfied when condition ii) is satisfied. We will summarize the discussion in this section as a corollary of theorem \ref{T:tm}:
\begin{corollary}
\label{C:tm1C}
   Let $V(\vec{\Phi})$ be a NHDM potential. Then $V$ is custodially symmetric, i.e.~is $SO(4)_C$-symmetric with a $SO(3)_C$-symmetric vacuum in a basis where the potential is manifestly $SO(4)_C$-symmetric, if and only if there is a Higgs basis transformation $U(R)\in U(N)$ where $R$ satisfies theorem \ref{T:tm}, and where $\langle 0| U(R)\vec{\Phi}|0 \rangle$ is a real vector.
\end{corollary}
Finally, the custodial symmetry will be spontaneously broken if the potential is  $SO(4)_C$-symmetric, but there is no Higgs basis where the potential is manifestly $SO(4)_C$-symmetric and where all VEVs are real at the same time. 

\subsection{Examples}
\label{sec:Examples}
\subsubsection{Applying a necessary condition}
\label{sec:ExamplesApplyingANecessaryCondition}
We will now give examples of potentials that are not $SO(4)_C$-symmetric, by applying the necessary condition \eqref{E:condComm2TestO_Simplified3HDM}. Consider a 3HDM potential $V$ where the quartic terms $V_E=K_a E_{ab} K_b$ in \eqref{E:V_Edef}
 are specified by
\begin{align}
	\label{E:EexNonCust3HDM}
   E=\begin{pmatrix} 
 &{\bigX}& 0 & 0  & 0 \\
& &  \vdots & \vdots  & \vdots \\
0 & \ldots & 0 &0 & 0 \\
0 & \ldots & 0 & 0 &0\\
0 & \ldots & 0 & 0 &0 
\end{pmatrix}.
\end{align}
 Let $X$ in \eqref{E:EexNonCust3HDM} be any real, symmetric and invertible $5\times 5$ matrix.
In the 3HDM with our choice of basis for $su(3)$ (see appendix \ref{sec:GeneralizedGellMannMatrices}), the column vector consisting of the bilinears $K_a$ becomes
\begin{align}
	\vec{K}=2\left(\widehat{C}_{12},\widehat{C}_{13},\widehat{C}_{23},\widehat{B}_{12},\widehat{B}_{13},\widehat{B}_{23},\frac{\widehat{A}_1-\widehat{A}_2}{2},\frac{\widehat{A}_1+\widehat{A}_2-2\widehat{A}_3}{2\sqrt{3}}\right)^T,
\end{align}
where the three first bilinears are not manifestly $SO(4)_C$-symmetric. 
 Since $X$ is invertible, its 5 columns are
linearly independent, and hence the 5 first columns of $E$ are linearly independent, while the 3 last columns of $E$ are zero.
Hence the dimension of the columnspace, and the rank, is 5, and the nullity (i.e.~the dimension of the nullspace) is $8-5=3=l$. Three orthonormal nullvectors of 
$E$ are then 
\begin{align}\label{E:nullv3HDMex}
	\vec{e}_6 &= \left(0,0,0,0,0,1,0,0\right)^T \equiv \tilde{n}_1, \nn \\
	\vec{e}_7 &= \left(0,0,0,0,0,0,1,0\right)^T \equiv  \tilde{n}_2,  \nn \\
	\vec{e}_8 &=\left(0,0,0,0,0,0,0,1\right)^T \equiv  \tilde{n}_3.
\end{align}
By inserting the nullvectors of \eqref{E:nullv3HDMex} into the necessary condition for the custodial symmetry for the 3HDM \eqref{E:condComm2TestO_Simplified3HDM} (alternatively, the more general necessary condition
\eqref{E:condComm2TestO} can be applied), we get
\begin{align}
	f^{123} &= \pm (\tilde{n}_1)_i (\tilde{n}_2)_j (\tilde{n}_3)_k f^{ijk}\nn \\
&= \pm \begin{pmatrix} 0 \\ \vdots \\
0 \\ 1 \\  0 \\  0
\end{pmatrix}_i 
\begin{pmatrix} 0 \\ \vdots \\
0 \\ 0 \\  1 \\  0 
\end{pmatrix}_j 
  \begin{pmatrix} 0 \\ \vdots \\
0 \\ 0 \\  0 \\  1
\end{pmatrix}_k f^{ijk}=\pm f^{678}, 
\end{align}
which does not hold, since $f^{123}=1/2$, while $f^{678}=0$, 
cf.~appendix \ref{sec:StructureConstantsForThe3HDM}. 
Hence a 3HDM potential containing quartic terms 
given by $V_E=K_a E_{ab} K_b$ and \eqref{E:EexNonCust3HDM}, is not $SO(4)_C$-symmetric.

We can control that the necessary condition \eqref{E:condComm2TestO_Simplified3HDM} makes sense, 
by checking that evidently $SO(4)_C$-symmetric terms satisfy the condition. Let the matrix 
$E$ of $V_E$ now be defined by
\ba\label{E:EcorrForm3HDM}
   E=\begin{pmatrix} 
 0 & 0 & 0 &\ldots & 0 \\
0 & 0& 0 &  \ldots &0 \\
0 & 0 & 0 &\ldots & 0 \\
\vdots &\vdots &\vdots & {\bigXdown}& \\
0 & 0 & 0 & & 
\end{pmatrix},
\ea
which is of the form \eqref{E:EcorrForm}, and hence give manifestly $SO(4)_C$-symmetric terms $V_E$. The block $X$ is again a real, symmetric and invertible, but otherwise arbitrary $5\times 5$ matrix.
Then the nullspace of $E$ has dimension 3, and three orthonormalized nullvectors are
\begin{align}\label{E:nullv3HDMex2}
	\vec{e}_1 &= \left(1,0,0,0,0,0,0,0\right)^T \equiv \tilde{n}_1, \nn \\
	\vec{e}_2 &= \left(0,1,0,0,0,0,0,0\right)^T \equiv \tilde{n}_2,  \nn \\
	\vec{e}_3 &=\left(0,0,1,0,0,0,0,0\right)^T \equiv \tilde{n}_3.
\end{align}
Now \eqref{E:condComm2TestO_Simplified3HDM} yields
\ba\label{E:condComm2TestOOk3HDM}
     f^{123}=\det(O) (\tilde{n}_1)_i (\tilde{n}_2)_j (\tilde{n}_3)_k f^{ijk}=\pm f^{123}, 
\ea
which holds as long as $O\in SO(3)$, and hence $V_E$ given by \eqref{E:EcorrForm3HDM} satisfy the necessary condition \eqref{E:condComm2TestO_Simplified3HDM} for $SO(4)_C$ symmetry, as it should.
Moreover, theorem \ref{T:tm} will hold with $R=I$ for quartic terms $V_E$ given by \eqref{E:EcorrForm3HDM}, and hence the terms $V_E$ are $SO(4)_C$-symmetric also by theorem \ref{T:tm}. Here the nullvectors \eqref{E:nullv3HDMex2} are the three first rows of $R=I$.
\subsubsection{A $SO(4)_C$-symmetric 3HDM potential}
\label{sec:ACustodialSymmetricPotentialInThe3HDM}
We will now apply theorem \ref{T:tm} to show that a certain 3HDM potential is $SO(4)_C$-symmetric. Consider the 3HDM potential given by
\begin{align}\label{E:full3HDMexE}
	 E= \left(
\begin{array}{cccccccc}
 0 & \frac{1}{\sqrt{2}} & 0 & 1 & 0 & 1 & \frac{\sqrt{\frac{3}{2}}}{2} &
   -\frac{1}{2 \sqrt{2}} \\
 \frac{1}{\sqrt{2}} & 1 & \frac{1}{\sqrt{2}} & 2 \sqrt{2} & -1 & 2 \sqrt{2} &
   \frac{\sqrt{3}}{2} & -\frac{1}{2} \\
 0 & \frac{1}{\sqrt{2}} & 0 & 1 & 0 & 1 & \frac{\sqrt{\frac{3}{2}}}{2} &
   -\frac{1}{2 \sqrt{2}} \\
 1 & 2 \sqrt{2} & 1 & -\frac{1}{2} & \frac{1}{\sqrt{2}} & -\frac{1}{2} & 0 & 0
   \\
 0 & -1 & 0 & \frac{1}{\sqrt{2}} & 0 & \frac{1}{\sqrt{2}} & 0 & 0 \\
 1 & 2 \sqrt{2} & 1 & -\frac{1}{2} & \frac{1}{\sqrt{2}} & -\frac{1}{2} & 0 & 0
   \\
 \frac{\sqrt{\frac{3}{2}}}{2} & \frac{\sqrt{3}}{2} &
   \frac{\sqrt{\frac{3}{2}}}{2} & 0 & 0 & 0 & \frac{3}{4} & -\frac{\sqrt{3}}{4}
   \\
 -\frac{1}{2 \sqrt{2}} & -\frac{1}{2} & -\frac{1}{2 \sqrt{2}} & 0 & 0 & 0 &
   -\frac{\sqrt{3}}{4} & \frac{1}{4} \\
\end{array}
\right),
\end{align}
and where the other parameters dependent of Higgs basis transformations are given by
\ba\label{E:3HDMexOtherParameters}
\vec{\xi}&&=\left(
\begin{array}{cccccccc}
 \frac{1}{6}, & -\frac{1}{3 \sqrt{2}}, & \frac{1}{6}, & -\frac{1}{6}, & -\frac{1}{3
   \sqrt{2}}, & -\frac{1}{6}, & \frac{1}{2 \sqrt{6}}, & -\frac{1}{6 \sqrt{2}} \\
\end{array}
\right)^T, \nn \\
\vec{\eta}&&=\left(
\begin{array}{cccccccc}
 \frac{5 \sqrt{2}}{3}, & -\frac{2}{3}, & \frac{5 \sqrt{2}}{3}, & -\frac{2 \sqrt{2}}{3}, &
   2, & -\frac{2 \sqrt{2}}{3}, & \frac{1}{\sqrt{3}}, & -\frac{1}{3} \\
\end{array}
\right)^T,
\ea
 confer \eqref{E:potentialAdjoint}. Other parameters are arbitrary.
  As we shall prove in this section, this potential is $SO(4)_C$-symmetric. First we check that the nullity of $E\geq N(N-1)/2=3$, which it is, 
since it has exactly 3 eigenvalues that equal zero. Mathematica then gives us the following orthonormal nullvectors of $E$,
\begin{align}\label{E:ExOrthonormNullv}
\tilde{n}_1 &=	 \left(0 , 0 , 0 , 0 , 0 , 0 ,  \frac{1}{2}, \frac{\sqrt{3}}{2}\right)^T, \nn \\
\tilde{n}_2 &=  \left(0 , 0 , 0 , -\frac{1}{\sqrt{2}} , 0 , \frac{1}{\sqrt{2}} , 0 , 0\right)^T, \nn \\
\tilde{n}_3 &= \left( -\frac{1}{\sqrt{2}} , 0 , \frac{1}{\sqrt{2}} , 0 , 0 , 0 , 0 , 0\right)^T.  
\end{align}
We then check that these nullvectors satisfy the necessary condition for custodial symmetry, 
for the 3HDM given by \eqref{E:condComm2TestO_Simplified3HDM}, and find that the nullvectors satisfy
condition \eqref{E:condComm2TestO_Simplified3HDM}, for $\det{O}=1$. 
Thus the necessary condition \eqref{E:condComm2TestO_Simplified3HDM}
also holds for all $SO(3)$ rotations of the nullvectors. We now try to solve the equations
\eqref{E:condComm3} in theorem \ref{T:tm} for the choice \eqref{E:ExOrthonormNullv} of nullvectors.  
As given in \eqref{E:NumberEqs} and the following discussion for the 3HDM, there will be 56 distinct (but possibly dependent) equations, including one that holds according to the already satisfied necessary condition. On the other hand, there are 40 variables (we do not include $SO(3)$ rotations of the nullvectors in this first attempt). Applying Mathemathica's Solve-command, then gives us two solutions:
When we include the already fixed three first rows of $R$, which consist of the (transposed) nullvectors \eqref{E:ExOrthonormNullv}, one solution $R_1$ that solves \eqref{E:condComm3} reads
\begin{align}
R_1=	\left(
\begin{array}{cccccccc}
 0 & 0 & 0 & 0 & 0 & 0 & \frac{1}{2} & \frac{\sqrt{3}}{2} \\
 0 & 0 & 0 & -\frac{1}{\sqrt{2}} & 0 & \frac{1}{\sqrt{2}} & 0 & 0 \\
 -\frac{1}{\sqrt{2}} & 0 & \frac{1}{\sqrt{2}} & 0 & 0 & 0 & 0 & 0 \\
 0 & -1 & 0 & 0 & 0 & 0 & 0 & 0 \\
 \frac{1}{\sqrt{2}} & 0 & \frac{1}{\sqrt{2}} & 0 & 0 & 0 & 0 & 0 \\
 0 & 0 & 0 & -\frac{1}{\sqrt{2}} & 0 & -\frac{1}{\sqrt{2}} & 0 & 0 \\
 0 & 0 & 0 & 0 & -1 & 0 & 0 & 0 \\
 0 & 0 & 0 & 0 & 0 & 0 & \frac{\sqrt{3}}{2} & -\frac{1}{2} \\
\end{array}
\right).
\end{align}
The other solution $R_2$ equals
\begin{align}
	R_2 = R^{\text{cc}} R_1 ,
\end{align}
where $R^{\text{cc}}$ is given by \eqref{E:Rcc}, and represents complex conjugation on the
Lie algebra $su(3)$ with our basis $\{i\lambda_j\}_{j=1}^8$ given in appendix \ref{sec:GeneralizedGellMannMatrices}. We now check that the parameters $\vec{\xi}$ and
 $\vec{\eta}$ given by \eqref{E:3HDMexOtherParameters} satisfy condition iii) of theorem \ref{T:tm}, which they do for both matrices $R_1$ and $R_2$. Hence the potential given by \eqref{E:full3HDMexE} and \eqref{E:3HDMexOtherParameters} is $SO(4)_C$-symmetric by theorem \ref{T:tm}.

Both solutions $R=R_1$ and $R_2$ correspond to the same, manifestly $SO(4)_C$-symmetric matrix $E'$:\footnote{The matrix $R_3= R_1 R^{\text{cc}} $ will also transform $E$ into a manifestly $SO(4)_C$-symmetric $E'$, but alters
the signs of the two first rows of $R_1$, and is hence not a solution when we have chosen the three first rows to equal the nullvectors in \eqref{E:ExOrthonormNullv}.}
\begin{align}
	E'=RER^T=\left(
\begin{array}{cccccccc}
 0 & 0 & 0 & 0 & 0 & 0 & 0 & 0 \\
 0 & 0 & 0 & 0 & 0 & 0 & 0 & 0 \\
 0 & 0 & 0 & 0 & 0 & 0 & 0 & 0 \\
 0 & 0 & 0 & 1 & -1 & 4 & -1 & -1 \\
 0 & 0 & 0 & -1 & 0 & -2 & 0 & 1 \\
 0 & 0 & 0 & 4 & -2 & -1 & 1 & 0 \\
 0 & 0 & 0 & -1 & 0 & 1 & 0 & 0 \\
 0 & 0 & 0 & -1 & 1 & 0 & 0 & 1 \\
\end{array}
\right).
\end{align}
According to theorem \ref{T:tm}, only one of the solutions $R_1$ and $R_2$ will correspond to a $SU(N)$ basis transformation of the Higgs fields. 
Here it turns out that $R_1$ corresponds to a $SU(3)$ transformation of the Higgs fields, and is hence a matrix in $Ad_{SU(3)}$.\footnote{By applying Mathematica's NSolve-command for $R_2$ and solving for $U$ via \eqref{E:defR}. The solution space then becomes empty.}  This also means that $R_1$ represents an inner automorphism of $su(3)$, while $R_2$ corresponds to an outer automorphism of $su(3)$, i.e.~an inner automorphism combined with complex conjugation.
A matrix $U_1\in SU(3)$ which corresponds to $R_1$ through \eqref{E:defR} is given by
\ba
U_1 =
\left(
\begin{array}{ccc}
 -\frac{1}{\sqrt{2}} & 0 & \frac{1}{\sqrt{2}} \\
 -\frac{i}{\sqrt{2}} & 0 & -\frac{i}{\sqrt{2}} \\
 0 & i & 0 \\
\end{array}
\right).
\ea

There will be two other $SU(3)$ matrices which correspond to $R_1$, namely $\alpha U_1$ and $\alpha^2 U_1$ where
$\alpha =e^{\frac{2\pi i}{3}}$, i.e.~a third-root of unity: Since the center of $SU(3)$ is $Z(SU(3))=\{\alpha I:\alpha^3=1\}\cong \N{Z}_3$, the kernel of the adjoint action of $SU(3)$ is $\text{Ker}(Ad)\cong \N{Z}_3$, and hence $Ad_{SU(3)}\cong SU(3)/ \N{Z}_3$. Then each
$R\in Ad_{SU(3)}$ will correspond to three different $U\in SU(3)$ through \eqref{E:defR}. The same is valid in $SU(N)$, that is, $Ad_{SU(N)}\cong SU(N)/ \N{Z}_N$, and each $R\in Ad_{SU(N)}$ correspond to exactly $N$ different $U\in SU(N)$ via \eqref{E:defR}.

\subsubsection{The Ivanov-Silva model}
\label{sec:IvanovSilva}
The Ivanov-Silva model was given as a counter-example to the then widely believed, but erroneous claim that an explicitly $CP$-invariant NHDM necessarily has a real basis \cite{Ivanov:2015mwl}. A $SO(4)_C$-symmetric potential
must be explicitly $CP$-invariant, since a $SO(4)_C$-symmetric potential can be written in a basis where all coefficients are real, which in our formalism means no terms linear in $\widehat{C}_a$ are present in the potential \cite{Olaussen:2010aq}. On the other hand, $CP$-invariance does not necessarily imply $SO(4)_C$ symmetry. 
The Ivanov-Silva model is an example of the latter. Since there is no real basis for this model, there will always be terms linear in
$\widehat{C}_a$ present in the potential, and hence it is not $SO(4)_C$-symmetric. The model's violation of the custodial symmetry is also implicitly stated in a footnote in \cite{Ivanov:2015mwl}, which says that the model has no other symmetries than powers of an order-4, generalized $CP$ transformation $J$,
defined by
\ba\label{E:Jgcp}
   J:  \Phi_m \to X_{mn}\Phi_n^\ast, 
\ea
where
\ba
 X= \left(   \begin{array}{ccc}
        1&0&0 \\
        0&0&i\\
        0&-i&0
    \end{array} \right).
\ea

We will now confirm the result that the Ivanov-Silva model is not $SO(4)_C$-symmetric, by applying theorem \ref{T:tm}:

The potential of Ivanov and Silva's 3HDM model can be written in terms of bilinears as $V=V_0+V_1$, 
 with
\ba
  V_0&=& -m^2_{11}\widehat{A}_1-m_{22}^2(\widehat{A}_2+\widehat{A}_3) 
  +\lambda_1  \widehat{A}_1^2 + \lambda_2(\widehat{A}_2^2+\widehat{A}_3^2)
  \nn \\
       &+& \lambda_3 \widehat{A}_1 (\widehat{A}_2+\widehat{A}_3) +\lambda_3'  \widehat{A}_2 \widehat{A}_3  \nn \\
       &+&\lambda_4(\widehat{B}_{12}^2+\widehat{C}_{12}^2+\widehat{B}_{13}^2+\widehat{C}_{13}^2) + \lambda_4'(\widehat{B}_{23}^2+\widehat{C}_{23}^2) ,
\ea
  and
\ba
   V_1 &=& 2\lambda_5 (\widehat{B}_{13}\widehat{B}_{12}- \widehat{C}_{13}\widehat{C}_{12})
   + \lambda_6 (\widehat{B}_{12}^2-\widehat{C}_{12}^2-\widehat{B}_{13}^2+\widehat{C}_{13}^2) \nn \\
   &+&2 \text{Re}(\lambda_8) (\widehat{B}_{23}^2-\widehat{C}_{23}^2)
   -4 \text{Im}(\lambda_8)\widehat{B}_{23}\widehat{C}_{23} \nn \\
   &+&2 \text{Re}(\lambda_9) (\widehat{A}_2 \widehat{B}_{23} 
   -\widehat{A}_3 \widehat{B}_{23})
  -2 \text{Im}(\lambda_9) (\widehat{A}_2 \widehat{C}_{23}- \widehat{A}_3 \widehat{C}_{23}),
\ea  
with all parameters real, except $\lambda_8$ and $\lambda_9$.
Moreover, $V_0$, $V_1$ and all the parameters are the same as in the original article \cite{Ivanov:2015mwl}. The Ivanov-Silva potential is then given by \eqref{E:potentialAdjoint} with the following parameters:

{\tiny \ba\label{E:EIvanov}
&&{E}= \\
&&\left(
\begin{array}{cccccccc}
 \frac{1}{4} \left(\lambda _4-\lambda _6\right) & -\frac{\lambda _5}{4} & 0 & 0 &
   0 & 0 & 0 & 0 \\
 -\frac{\lambda _5}{4} & \frac{1}{4} \left(\lambda _4+\lambda _6\right) & 0 & 0 &
   0 & 0 & 0 & 0 \\
 0 & 0 & \frac{1}{4} \left(\lambda _{4}'-2 \Re\left(\lambda _8\right)\right) & 0
   & 0 & -\frac{1}{2} \Im\left(\lambda _8\right) & \frac{\Im\left(\lambda
   _9\right)}{4} & -\frac{1}{4} \sqrt{3} \Im\left(\lambda _9\right) \\
 0 & 0 & 0 & \frac{1}{4} \left(\lambda _4+\lambda _6\right) & \frac{\lambda
   _5}{4} & 0 & 0 & 0 \\
 0 & 0 & 0 & \frac{\lambda _5}{4} & \frac{1}{4} \left(\lambda _4-\lambda
   _6\right) & 0 & 0 & 0 \\
 0 & 0 & -\frac{1}{2} \Im\left(\lambda _8\right) & 0 & 0 & \frac{1}{4} \left(2
   \Re\left(\lambda _8\right)+\lambda _{4}'\right) & -\frac{1}{4}
   \Re\left(\lambda _9\right) & \frac{1}{4} \sqrt{3} \Re\left(\lambda _9\right)
   \\
 0 & 0 & \frac{\Im\left(\lambda _9\right)}{4} & 0 & 0 & -\frac{1}{4}
   \Re\left(\lambda _9\right) & \frac{1}{4} \left(\lambda _1+\lambda _2-\lambda
   _3\right) &E_{78}
   \\
 0 & 0 & -\frac{1}{4} \sqrt{3} \Im\left(\lambda _9\right) & 0 & 0 & \frac{1}{4}
   \sqrt{3} \Re\left(\lambda _9\right) &E_{78} &E_{88} \\
\end{array}
\right)   \nn
\ea }%
 with $E_{78}= \frac{\lambda _1-\lambda _2-\lambda
   _3+\lambda _{3}'}{4 \sqrt{3}}$, $E_{88}= \frac{1}{12} \left(\lambda _1+5 \lambda
   _2-\lambda _3-2 \lambda _{3}'\right)$ and where $\Re(\lambda)=\text{Re}(\lambda)$, while $\Im(\lambda)=\text{Im}(\lambda)$. The other parameters of the Ivanov-Silva potential are
\ba
   \eta_0&=&\frac{1}{6} \left(\lambda _1+2 \lambda _2+2 \lambda _3+\lambda _{3}'\right), \nn \\
\eta_1&=&\eta_2=\eta_3=\eta_4=\eta_5=\eta_6=0, \nn \\
\eta_7&=&\frac{1}{24} \left(2 \sqrt{6} \lambda _1-2 \sqrt{6} \lambda _2+\sqrt{6} \lambda
   _3-\sqrt{6} \lambda _{3}'\right), \nn \\
\eta_8&=& \frac{\eta_7}{\sqrt{3}}. 
\ea
The matrix $E$ given by \eqref{E:EIvanov} is generally not custodially symmetric, since it is generally not singular, and hence does not satisfy condition i) of theorem \ref{T:tm}, which
requires an eigenvalue $0$ of (at least) multiplicity $3$. To see this, set all parameters e.g.~to $1$, 
 with the result 
\ba\label{E:EIvanovNum}
E=
\left(
\begin{array}{cccccccc}
 0 & -\frac{1}{4} & 0 & 0 & 0 & 0 & 0 & 0 \\
 -\frac{1}{4} & \frac{1}{2} & 0 & 0 & 0 & 0 & 0 & 0 \\
 0 & 0 & -\frac{1}{4} & 0 & 0 & 0 & 0 & 0 \\
 0 & 0 & 0 & \frac{1}{2} & \frac{1}{4} & 0 & 0 & 0 \\
 0 & 0 & 0 & \frac{1}{4} & 0 & 0 & 0 & 0 \\
 0 & 0 & 0 & 0 & 0 & \frac{3}{4} & -\frac{1}{4} & \frac{\sqrt{3}}{4} \\
 0 & 0 & 0 & 0 & 0 & -\frac{1}{4} & \frac{1}{4} & 0 \\
 0 & 0 & 0 & 0 & 0 & \frac{\sqrt{3}}{4} & 0 & \frac{1}{4} \\
\end{array}
\right).
\ea
Then $\det(E) =1/65536$, and $E$ is not singular, and hence does not correspond to a $SO(4)_C$-symmetric potential.

Now expand the scalar fields around the vacuum, i.e.~let
\ba
   \Phi_j = \left(   \begin{array}{c} \phi_j^+ \\ \frac{1}{\sqrt{2}}  (v_j+\eta_j+ i \chi_j) \end{array}     \right),
\ea
 where $v_1$ is real and $v_2=v_3=0$ in the Ivanov-Silva model.
 Then the neutral mass matrix split in an $\eta$-sector and a $\chi$-sector, where the (tree-level) masses 
of the $\eta$'s are  
\ba
    m_{\eta_1}=2 m_{11}^2, \quad m_{\eta_2,\eta_3}=\frac{1}{2} \left(-2 m_{22}^2+ v_1^2(\lambda _3+\lambda _4) \mp v_1^2\sqrt{\lambda _5^2 +\lambda _6^2 }\right),
\ea 
and where $\eta_1$ is the SM Higgs.  The masses 
of the $\chi$'s are given by
\ba\label{E:ISmchi}
    m_{\chi_1}=0, \quad m_{\chi_2,\chi_3}=m_{\eta_2,\eta_3},
\ea 
and finally the charged masses are 
\ba
    m_{\phi^\pm_1}=0,\quad m_{\phi^\pm_2,\phi^\pm_3}=\frac{\lambda _3 v_1^2}{2}- m_{22}^2,
\ea
where the two non-zero masses are identical.

 The mass degeneration between the $\eta$-sector and the $\chi$-sector, given by \eqref{E:ISmchi}, is not the same mass degeneration that occurs in custodially symmetric potentials (i.e.~potentials where the VEVs are real in some basis where $SO(4)_C$ symmetry is manifest, cf.~corollary \ref{C:tm1C}) \cite{Olaussen:2010aq}: The latter mass degeneration is between the $\chi$- ($CP$-odd) sector and the charged sector, where the two sectors get identical masses. In contrast, the mass degeneration in the Ivanov-Silva model is only partial, and between other sectors. This mass degeneration is caused by the generalized $CP$ symmetry $J$ of the model, confer \eqref{E:Jgcp}, which both the potential and the vacuum respect.

\subsubsection{The necessary condition is not sufficient}
\label{sec:Thenecessaryconditionisnotsufficient}
 We will now conclude with a 3HDM example which shows the necessary condition for the custodial symmetry, given by  \eqref{E:condComm2TestO_Simplified3HDM}, 
as expected is not sufficient: 
The following nullvectors satisfy the necessary condition \eqref{E:condComm2TestO_Simplified3HDM},
\begin{align}\label{E:moteks}
\tilde{n}_1 &=	 \left(1 , 0 , 0 , 0 , 0 , 0 , 0, 0 \right)^T, \nn \\
\tilde{n}_2 &=  \left(0 , 0 , 0 , 0 , 0 ,0 , 1 , 0\right)^T, \nn \\
\tilde{n}_3 &= \left( 0 , 0 ,0 , \frac{1}{2} , \frac{\sqrt{3}}{2}  , 0 , 0 , 0\right)^T.  
\end{align}
Furthermore, if we apply Mathematica's NSolve-command to solve the full set of equations \eqref{E:condComm3}, with the nullvectors \eqref{E:moteks} as the three first rows of $R$, we get a solution space which is empty. Thus the necessary condition \eqref{E:condComm2TestO_Simplified3HDM} for $SO(4)_C$ symmetry, is not sufficient.

\section{Summary}
\label{sec:Summary}
	We started by organizing the NHDM bilinears that vary under $SU(N)$ Higgs basis transformations in a vector $K_a$, given by \eqref{E:Kalpha}, by putting the $N(N-1)/2$ custodial symmetry-violating bilinears first in $K_a$. These custodial symmetry-violating bilinears are (proportional to) bilinears of the type $\widehat{C}$, see \eqref{E:custViolBilinsInKa} and \eqref{E:bilinears}. We derived a Higgs basis-invariant necessary condition \eqref{E:condComm2TestO} for $SO(4)_C$ symmetry, that only involves the nullvectors of the matrix $E$ from the quartic part of the NHDM potential \eqref{E:potentialAdjoint}. In the case of the 3HDM with $\text{Nullity}(E)=3$, this necessary condition was simplified to \eqref{E:condComm2TestO_Simplified3HDM}, since if it holds for one choice of the nullvectors, it holds for all $SO(3)$ rotations of the nullvectors. The main result of this article, theorem \ref{T:tm}, gave us basis-invariant necessary and sufficient conditions for an explicit custodial symmetry in a general NHDM, $N \geq 3$. Corollary \ref{C:tm1C} in section \ref{sec:SpontaneousBreakingCustodialSymmetry} yielded corresponding conditions for a simultaneously custodially symmetric potential and vacuum. In section \ref{sec:Examples} we applied both the necessary condition and the necessary and sufficient conditions of theorem \ref{T:tm} on some 3HDM potentials. Here we showed that a certain family of 3HDM potentials, given by \eqref{E:EexNonCust3HDM} are not $SO(4)_C$-symmetric, since they do not fulfil the necessary condition \eqref{E:condComm2TestO_Simplified3HDM}. Moreover, we showed that another 3HDM potential, given by \eqref{E:full3HDMexE} and \eqref{E:3HDMexOtherParameters}, is $SO(4)_C$-symmetric since it fulfils the necessary and sufficient conditions given by theorem \ref{T:tm}. Finally, in section \ref{sec:IvanovSilva}, we also applied our methods to demonstrate that the Ivanov-Silva model is not $SO(4)_C$-symmetric.
	\appendix

	\section{Factorizable quartic terms}
\label{sec:FactorizableQuarticTerms}
Conditions for $SO(4)_C$ in the quartic terms of the NHDM potential can be formulated relatively easy if the quartic terms are factorizable. We will define the quartic terms to be factorizable if they can be written
\begin{align}\label{E:factorizable}
	V_\text{4}= \vec{\Phi}^\dag A \vec{\Phi}\; \vec{\Phi}^\dag B \vec{\Phi},
\end{align}
where $A$ and $B$ are Hermitian $N\times N$ matrices, and $\vec{\Phi}$ is the $N\times 1$ vector consisting of the $N$ Higgs doublets,
\begin{align}
	\vec{\Phi}=(\Phi_1,\Phi_2,\ldots,\Phi_N)^T.
\end{align}
 The general NHDM potential cannot be written this way, since $A$ and $B$ together contain $2N^2$ free parameters, while the quartic part of the general NHDM potential contains $\frac{1}{2}N^2(N^2+1)$ free parameters, which supersedes $2N^2$ for $N>1$.
 
Let $\mu$ be the (Hermitian) mass matrix of the potential, i.e.~write the quadratic terms of the potential as
\begin{align}
	V_\text{2}= \vec{\Phi}^\dag \mu \vec{\Phi}.
\end{align}
Then, if the quartic terms are factorizable as in \eqref{E:factorizable}, the
potential is $SO(4)_C$-symmetric if and only if there is a basis transformation 
\begin{align}\label{E:basisTrafo}
\vec{\Phi} \to \vec{\Phi}'=U\vec{\Phi}, \quad U\in SU(N),
\end{align}
such that the matrices $U^\dag A U$, $U^\dag B U$ and $U^\dag \mu U$ are simultaneously real:
The $SO(4)_C$-violating terms are terms involving one or two factors of bilinears of type 
$\widehat{C}$, and the bilinears of type $\widehat{C}$ are generated by the imaginary parts of the matrices $A$, $B$ and $\mu$.    

We will now give criteria for when a family of Hermitian matrices can be made simultaneously real
by similarity transformations. 
Let $\{A_i\}_{i=1}^k$ be a family of Hermitian $N\times N$ matrices. Then there is a 
$U\in SU(N)$ that simultaneously makes the matrices $\{A_i\}_{i=1}^k$ similar to real matrices, that is
\begin{align}
	U^\dag A_i U \quad \text{is real} \quad \forall i,\; 1\leq i \leq k.
\end{align}
if and only if there is a symmetric $W\in SU(N)$ such that
\begin{align}\label{E:FactorCriteria}
	A_i W=W A_i^T \quad \forall i,\; 1\leq i \leq k.
\end{align}
Proof. ($\Rightarrow$): Assume $U^\dag A_i U$ is real for all $i$, where $U\in SU(N)$. Then
 $U^\dag A_i U=(U^\dag A_i U)^\dag=(U^\dag A_i U)^T=U^TA_i^TU^\ast $. Hence  
 $A_i UU^T =UU^TA_i^T$ since $U^\ast U^T=I$. Let $W=UU^T$, and $W$ is unitary and symmetric.

($\Leftarrow$): If $W$ is unitary and symmetric, there exists an unitary, symmetric matrix $U$, such that $U^2=W$ \cite{Zhang:mt}. 
Then $A_i U U^T = U U^T A_i^T$ for all $i$ by the assumption \eqref{E:FactorCriteria}. Hence $U^\dag A_i U = U^T A_i^T U^\ast$,
which infers that $(U^\dag A_i U)^T=U^T A_i^T U^\ast =U^\dag A_i U$, which means $U^\dag A_i U$ is real
for all $i$.

This leads us to the following sufficient and necessary conditions for having $SO(4)_C$ symmetry in a potential with factorizable quartic terms: The potential is $SO(4)_C$-symmetric if and only if there is a symmetric $W\in SU(N)$ such that
\begin{align}\label{E:FactorCriteria2}
	X W=W X^T \quad \forall X \in \{A, B, \mu\}. 
\end{align}
Eq.~\eqref{E:FactorCriteria2} represents a set of $3N^2$ linear equations in 
$N^2+N$ variables, if we disregard the condition $W\in SU(N)$. Including the latter condition,
 eq.~\eqref{E:FactorCriteria2} consists of $3N^2$ non-linear equations in $N^2-1$ (the dimension of $SU(N)$) variables. In any case the set of equations is overdetermined, and generally have no solution, which reflects that the potential generally is not custodially symmetric. 
	
	\section{A basis for $su(N)$}
	\label{sec:GeneralizedGellMannMatrices}
	We will now define a basis $\{v_a\}_{j=1}^{N^2-1}=\{i\lambda_a\}_{j=1}^{N^2-1}$ for $su(N)$, appropriate for our purposes. The Lie algebra $su(N)$ consists of anti-Hermitian $N\times N$ matrices, i.e.~matrices $A$ with the property $A^\dag=-A$. Generalized Gell-Mann matrices are on the other hand Hermitian.
	The matrices $\lambda_a$ in our basis will be the same as the generalized Gell-Mann matrices given in e.g.~\cite{Maniatis:2015gma}, but their order will be different. We will order all the imaginary matrices first, corresponding to the $SO(4)_C$-violating bilinears $\widehat{C}$. 
	Let
	\begin{align}
		\vec{e}_1&=\left(1,0,\ldots,0\right)^T, \nn \\
		         &\vdots \nn \\
			\vec{e}_N&=\left(0,\ldots,0,1\right)^T.			
	\end{align}
	Then the $N\times N$ matrix $u=\vec{e}_m \vec{e}_n^{\,\dag}$ has elements $u_{mn}=1$ for fixed $m$ and $n$, and	$u_{kl}=0$ for all $(k,l)\ne (m,n)$.

Now let $1\leq m < n \leq N$, and $a=a(m,n)$ be defined as in \eqref{E:encoding}, and let
\begin{align}\label{E:CandB}
	\lambda_a &= -i \vec{e}_m \vec{e}_{n}^{\,\dag} + i \vec{e}_{n} \vec{e}_{m}^{\,\dag} \quad &\text{for}
	\quad  a = a(m,n), \nn \\
   \lambda_b &=  \vec{e}_m \vec{e}_{n}^{\,\dag} + \vec{e}_{n} \vec{e}_{m}^{\,\dag}  \quad &\text{for}
	\quad  b = a(m,n) +\frac{N(N-1)}{2}.   
\end{align}
	Here we get
	\begin{align}
		K_c=\text{Tr}(\tilde{K}\lambda_c)= \vec{\Phi}^\dag \lambda_c \vec{\Phi},
	\end{align}
	where
	\begin{align}
		K_a &= 2 \widehat{C}_{m(a),n(a)}  \quad \text{for} \quad 1\leq a \leq \frac{N(N-1)}{2}\equiv k, \\
		K_b &= 2 \widehat{B}_{m(a),n(a)}  \quad \text{for} \quad k+1 \leq b=k+a(m,n)  \leq 2k.
	\end{align}
	The bilinears $\widehat{C}$ and $\widehat{B}$ are hence ordered "lexicographically", e.g.~$\widehat{C}_{12},\widehat{C}_{13},\ldots,\widehat{C}_{1 N},\widehat{C}_{23},\widehat{C}_{24},\ldots, \widehat{C}_{N-1,N}$, and afterwards the bilinears of type $\widehat{B}$ in the same pattern. 
	
	Finally, we define the diagonal (and traceless) matrices,
	\begin{align}\label{E:linCombsA}
	\lambda_j =	\sqrt{\frac{2}{m(m+1)}}\left[\sum_{l=1}^m \vec{e}_{l} \vec{e}_{l}^{\,\dag} - m \vec{e}_{m+1} \vec{e}_{m+1}^{\,\dag}\right], 
	\end{align}
	{for} $1\leq m \leq N-1$ {and} $j =m+N(N-1)$.
	The bilinears $K_j=\text{Tr}(\tilde{K}\lambda_j)$ of \eqref{E:linCombsA} are different linear
	combinations of the bilinears $\widehat{A}_n$, cf.~\eqref{E:bilinears}, orthogonal to $K_0=\sqrt{\frac{2}{N}} \,\sum_{j=1}^{N}\widehat{A}_j$. We denote the matrices $\{\lambda_c\}_{c=1}^{N^2-1}$ constructed above, the alternatively ordered, generalized Gell-Mann matrices.
	
 \subsection{Structure constants for the 3HDM}
	\label{sec:StructureConstantsForThe3HDM}
	The structure constants $f^{ijk}$ of the alternative, generalized Gell-Mann matrices $\{\lambda_c\}_{c=1}^{N^2-1}$ are given by	
	\begin{align}
		\left[\lambda_i,\lambda_j \right]=f^{ijk}\lambda_k.
	\end{align}
	The structure constants corresponding to $su(3)$, relevant for the 3HDM, with the alternative Gell-Mann matrices given by \eqref{E:CandB} and \eqref{E:linCombsA}
	will be the same as the structure constants of the ordinary Gell-Mann matrices, but the indices will be changed due to the change of the ordering of the matrices $\lambda_c$. Our permutation of the
	original Gell-Mann matrices is $(4,1,7,5,2,6,3,8)$, which means that the ordinary Gell-Mann matrix $\lambda_1$ is $\lambda_4$ in our alternative order, and so on.
	The structure constants then become changed according to the permutation of the Gell-Mann matrices, which means that for the matrices of \eqref{E:CandB} and \eqref{E:linCombsA} the following hold:
	\begin{align}\label{E:structureCons3HDM}
		f^{147}&=-1, \quad f^{123}=f^{156}=f^{246}=f^{275}=f^{345}=f^{367}= \frac{1}{2}, \nn \\
		f^{285}&=f^{386}=\frac{\sqrt{3}}{2}.
	\end{align}
	The structure constant $f^{abc}$ is completely antisymmetric in all its indices, i.e.~interchanging two indices changes the sign of 
	$f$ \cite{metha}. All structure constants not derivable from \eqref{E:structureCons3HDM} through permutations of the indices, are zero.
	

\section{The inclusion $Ad_{SU(N)}\subset SO(N^2-1)$}
\label{sec:WhenisAd_{SU(N)}subsetSO(N^2-1)}
When the basis vectors of a Lie algebra are orthogonal relative to the inner product induced by the Killing form, the matrix $R(U)$ will be orthogonal. The Killing form $\kappa$ is defined $\kappa(X,Y) =\text{tr}(ad(X)ad(Y))$, where $ad(X)(Y)=[X,Y]$, and is a linear transformation (a matrix) on the Lie algebra $\mathfrak{g}$. Let $b=\{v_j\}_{j=1}^{N^2-1}$ be a basis for $\mathfrak{g}$. The Killing form is invariant under automorphisms of the Lie algebra $\mathfrak{g}$, i.e. $\kappa(r(X),r(Y))=\kappa(X,Y)$ for all $X,Y\in \mathfrak{g}$ for all automorphisms $r\in \text{Aut}(\mathfrak{g})$. Then, if 
\ba
\kappa(v_j,v_k)\propto \delta_{jk}, 
\ea
the invariance of $\kappa$ under $r$ for general $X,Y$ induce 
\ba
R_{kl}R_{km}=\delta_{lm},
\ea
 which means $R$ is orthogonal, where $R$ is the matrix associated with the automorphism $r$, given the basis $b$. For $\mathfrak{g}=su(N)$, the Killing form is 
\ba
\kappa(X,Y)=2N\text{tr}(XY).  
\ea
If $\mathfrak{g}=su(N)$ and we choose $b$ to be a basis generated by generalized Gell-Mann matrices, a "Gell-Mann basis" $b=\{i \lambda_j\}_{j=1}^{N^2-1}$, we get 
\ba
\text{tr}(i\lambda_j i\lambda_k)=-2\delta_{jk}, 
\ea
and hence the matrices of $Ad_{SU(N)}$ become orthogonal with our preferred basis $b$. Moreover, $R(I)=I$, $R(U)$ is a continuous function of $U$ and $SU(N)$ is connected, hence all matrices $R(U)$ will be contained in the identity component of the orthogonal group $O(N^2-1)$. Thus 
\ba
Ad_{SU(N)}\subset SO(N^2-1).
\ea 
As indicated by the above, there are bases where $Ad_{SU(N)}$ does not consist of orthogonal matrices. Consider a general change of Lie algebra basis from a Gell-Mann basis to another basis, given by $i\lambda_a \to i\lambda_a'=M_{ab}i\lambda_b$, where $M$ is a real and invertible matrix. Then the matrices of $Ad_{SU(N)}$ relative to the primed basis will be given by $R_{da}'(U)=M_{dc}R_{cb}M^{-1}_{ba}$, by applying \eqref{E:defR}. By choosing $M$ with a determinant which differs from $\pm 1$, we easily get examples of matrices $R'(U)\in Ad_{SU(N)}'$ not being orthogonal.	
\section{Complex conjugation as an automorphism of $su(N)$}
	\label{sec:CCisAnOuterAutomorphismOfsu(N)}
	Complex conjugation is an inner automorphism of $su(2)$ while it is an outer automorphism of $su(N)$, $N>2$: For $N=2$, complex conjugation will be implemented by a similarity transformation with \ba
U=
\begin{pmatrix}
    0 & -1 \\
    1 & 0 
 \end{pmatrix}. 
\ea
The outer automorphism group of $su(N)$ is isomorphic to the automorphism group of the Dynkin diagram $A_{N-1}$, which is $\N{Z}_2$ for $N>2$
and the trivial group $\N{Z}_1$ for $N=2$. 
This means  
 $\text{Out}(su(N))$ for $N>2$ consists of only one non-trivial element, namely complex conjugation: Complex conjugation is obviously a $\N{R}$-linear bijection on $su(N)$, and respects the commutator, so complex conjugation is an automorphism on $su(N)$. To see that complex conjugation never equals a (unitary) similarity transformation for $N>2$, consider the diagonal matrix
\ba
X_3=i\cdot \text{diag}(1,1-2)\in su(3).
\ea
 Assume that complex conjugation is a similarity transformation. Then there exists an $U\in SU(3)$, such that $UX_3 U^\dag =X_3^\ast$ holds. But similar matrices have the same determinant, and the determinant of $X_3$ and $X_3^\ast$ differs by a factor $(-1)$, and hence they can not be similar, and thus complex conjugation cannot be an inner automorphism on $su(3)$.

 For $su(N)$, $N>3$, assume again that complex conjugation on $su(N)$ equals a similarity transformation, and consider the $N\times N$ matrix
\ba
   X= 
	 \left(\begin{array}{cc}
		 X_3 & 0 \\
		  0 & 0
	 \end{array}\right),
\ea
 where $X_3$ is defined above and $X\in su(N)$ (it is anti-Hermitian).
By assumption, there should exist an $U\in SU(N)$ such that $Y^\ast=UYU^\dag$ for all $Y\in su(N)$.
The characteristic polynomial of $X$ is given by $\det(tI-X)=\det(tI_{3\times 3}-X_3)\cdot \det(tI_{N-3\times N-3})=\det(tI_{3\times 3}-X_3)\cdot t^{N-3}$, and by the same manner 
 the characteristic polynomial of $X^\ast$ is $\det(tI-X^\ast)=\det(tI_{3\times 3}-X_3^\ast)\cdot t^{N-3}$.
We then
calculate the difference between the characteristic polynomials of the, by assumption similar, matrices $X$ and $X^\ast$,
 \ba
  \det(tI-X)-\det(tI-X^\ast) =-4i t^{N-3},
\ea
a contradiction, since similar matrices should have the same characteristic polynomial.
This again means that complex conjugation is not an inner automorphism for $su(N)$.
	
\section*{Acknowledgments}
Dedicated to Marit Julie Aase.
\newline
\newline
The author also wishes to thank K.~Skotheim, M.~Kachelrie\ss, S.~Willenbrock, M.~Zhang, C.~C.~Nishi, E.~Straume, P.~Osland and R.~K.~Solberg for helpful communication.

\end{document}